 \documentclass[journal]{IEEEtran}

\usepackage{algorithmic}

\usepackage{hyperref}
\let\svthefootnote\thefootnote
\newcommand\freefootnote[1]{%
  \let\thefootnote\relax%
  \footnotetext{#1}%
  \let\thefootnote\svthefootnote%
}

\usepackage[T1]{fontenc}

\usepackage{graphicx}
\usepackage[noadjust]{cite}
\usepackage{mcite}
\usepackage{amsfonts,helvet}
\usepackage{fancyhdr}
\usepackage{threeparttable}
\usepackage{epsf,epsfig}
\usepackage{amsthm}
\usepackage{amsmath}
\usepackage{siunitx}
\usepackage{amssymb}
\usepackage{dsfont}
\usepackage{subfigure}
\usepackage{color}
\usepackage{enumerate}
\usepackage{hyperref}
\usepackage{cancel}
\usepackage{bbm}
\usepackage{dsfont}
\usepackage[subnum]{cases}
\usepackage{adjustbox}
\usepackage[linesnumbered,ruled]{algorithm2e}
\usepackage{multicol}
\usepackage[english]{babel}
\usepackage{amsmath}
\usepackage{boldline}
\usepackage{graphicx}

\newtheorem{theorem}{Theorem}

\newtheorem{corollary}{Corollary}

\newtheorem{proposition}{Proposition}
\newtheorem{remark}{Remark}

\newmuskip\pFqmuskip

\def\be{{\bf e}}
\def\bff{{\bf f}}
\def\bg{{\bf g}}
\def\bh{{\bf h}}

\def\bn{{\bf n}}

\def\bq{{\bf q}}

\def\bs{{\bf s}}

\def\bu{{\bf u}}

\def\bw{{\bf w}}
\def\bx{{\bf x}}
\def\by{{\bf y}}

\def\bA{{\bf A}}

\def\bC{{\bf C}}
\def\bD{{\bf D}}
\def\bE{{\bf E}}

\def\bG{{\bf G}}
\def\bH{{\bf H}}
\def\bI{{\bf I}}

\def\bK{{\bf K}}

\def\bM{{\bf M}}

\def\bR{{\bf R}}

\def\bW{{\bf W}}

\def\bZ{{\bf Z}}

\def\cC{\mbox{$\mathcal{C}$}}

\def\cL{\mbox{$\mathcal{L}$}}

\def\cN{\mbox{$\mathcal{N}$}}

\def\bbC{\mbox{$\mathbb{C}$}}

\def\bbE{\mbox{$\mathbb{E}$}}

\def\bbR{\mbox{$\mathbb{R}$}}

\def\ubg{\mbox{$\underline{\bf{g}}$}}

\def\ubw{\mbox{$\underline{\bf{w}}$}}

\def\btau{\boldsymbol{\tau}}
\def\bSigma{\boldsymbol{\Sigma}}

\def\bLambda{{\pmb\Lambda}}

\usepackage{array}

\makeatletter
\newcommand{\thickhline}{%
    \noalign {\ifnum 0=`}\fi \hrule height 1pt
    \futurelet \reserved@a \@xhline
}
\newcolumntype{"}{@{\hskip\tabcolsep\vrule width 1pt\hskip\tabcolsep}}

\IEEEoverridecommandlockouts

\title{
Coordinated Per-Antenna Power Minimization for Multicell Massive MIMO Systems with Low-Resolution Data Converters}
\author{
Yunseong Cho, {\it Graduate Student Member, IEEE,} Jinseok Choi, {\it Member, IEEE,}  \\
and Brian L. Evans, {\it Fellow, IEEE} \thanks{

Y. Cho and B. L. Evans are with the 6G@UT Research Center, Wireless Networking and Communication Group (WNCG), Dept. of Electrical and Computer Engineering, The University of Texas at Austin, Austin, TX 78701 USA. (e-mail: yscho@utexas.edu, bevans@ece.utexas.edu).

J. Choi is with School of Electrical Engineering, Korea Advanced Institute of Science and Technology (KAIST), Daejeon, South Korea
(e-mail: jinseok@kaist.ac.kr).
}
\thanks{
This work was supported in part by NVIDIA and AT\&T Labs, affiliates of the 6G@UT Research Center at UT Austin, by the National Research Foundation of Korea (NRF) grant funded by the Korea government (MSIT) (No. RS-2023-00219443), and by the MSIT (Ministry of Science and ICT), Korea, under the ITRC (Information Technology Research Center) support program (IITP-2023-RS-2023-00259991) supervised by the IITP (Institute for Information \& Communications Technology Planning \& Evaluation).
}
}

\begin{document}

\maketitle
\begin{abstract}
\freefootnote{This work was presented in part at the {\em IEEE Wireless Communications and Networking Conference (IEEE WCNC)}, Austin, Texas, April 2022 \cite{cho2022coordinated}.}
A multicell-coordinated beamforming solution for massive multiple-input multiple-output orthogonal frequency-division multiplexing (OFDM) systems is presented when employing low-resolution data converters and per-antenna level constraints.
For a more realistic deployment, we aim to find the downlink (DL) beamformer that minimizes the maximum power on transmit antenna array of each basestation under received signal quality constraints while minimizing per-antenna transmit power.
We show that strong duality holds between the primal DL formulation and its manageable Lagrangian dual problem which can be interpreted as the virtual uplink (UL) problem with adjustable noise covariance matrices.
For a fixed set of noise covariance matrices, we claim that the virtual UL solution is effectively used to compute the DL beamformer and noise covariance matrices can be subsequently updated with an associated subgradient.
Our primary contributions are then (1) formulating the quantized DL OFDM  antenna power minimax problem and deriving its associated dual problem,
(2) showing strong duality and interpreting the dual as a virtual quantized UL OFDM problem, and
(3) developing an iterative minimax algorithm based on the dual problem.
Simulations validate the proposed algorithm in terms of the maximum antenna transmit power and peak-to-average-power ratio.

\end{abstract}
\begin{IEEEkeywords}
Multicell-coordinated beamforming, transmit power minimax problem, uplink-downlink duality, low-resolution data converters, per-antenna power constraint.
\end{IEEEkeywords}

\section{Introduction}
\label{sec:intro}


Massive multiple-input-multiple-output (MIMO) has drawn attention for  future wireless communication systems because of its remarkable gain in spectral efficiency and capacity \cite{marzetta2010noncooperative}.
However, employing the massive number of power-hungry high-resolution analog-to-digital converters (ADCs) and digital-to-analog converters (DACs) results in prohibitively high power consumption.
Since the power consumption of data converters experiences exponential growth with an increase in the number of quantization bits \cite{walden1999analog,lee2008analog}, the utilization of low-resolution data converters emerges as a crucial strategy for reducing both hardware expenses and power usage of the radio frequency (RF) chains in the context of large antenna arrays \cite{ribeiro2018energy}.
Accordingly, the adoption of low-resolution quantizers has gathered momentum as a promising power-efficient alternative and has been widely investigated \cite{choi2016near,studer2016quantized,choi2019robust,cho2019one,choi2017resolution,  choi2019two,choi2020quantized,nguyen2021linear,yuan2020toward}.

In addition to power consumption, interference has emerged as a critical consideration in modern wireless systems.
Consequently, intra-cell and inter-cell interference as well as quantization error must be carefully considered when analyzing and designing power-efficient multicell communication networks to achieve a desired performance.
Moreover, for practical implementations, it is desirable to impose a per-antenna power constraint that restricts the transmit power of each antenna because the communication system can operate with more energy-efficient power amplifiers and prevent nonlinear distortion \cite{yu2007transmitter, dahrouj2010coordinated}.
In this regard, we investigate coordinated multipoint (CoMP) beamforming (BF) and power control (PC) problems in multicell and multiuser massive MIMO orthogonal frequency-division multiplexing (OFDM) systems with low-resolution data converters and per-antenna level power and quality-of-service constraints.



\subsection{Prior Work}
Low-resolution ADC architectures have been the subject of extensive research in recent years to provide power-efficient communications \cite{choi2016near, studer2016quantized,wang2017bayesian,wang2019reliable,choi2019robust,jacobsson2017throughput, zhang2016mixed, choi2017resolution, orhan2015low, xu2019uplink, choi2019two}.
In order to properly handle the severe nonlinearities in low-resolution ADCs, many studies have re-engineered essential wireless communication functions such as channel estimation and data detection \cite{choi2016near, studer2016quantized,wang2017bayesian,wang2019reliable, choi2019robust, jacobsson2017throughput}.
Since low-resolution data converters destroy the orthogonality of subcarriers in OFDM systems, novel OFDM channel estimation and symbol detection were developed and integrated into a turbo framework, justifying its feasibility and reliability in over-the-air experiments when using low-resolution ADCs
\cite{wang2017bayesian,wang2019reliable}.
In \cite{choi2016near}, a near maximum likelihood channel estimation and symbol detection were proposed for 1-bit ADCs while showing improved estimation accuracy compared to expectation-maximum estimators.
Maximum a-posteriori detection and channel
estimation with low-resolution ADCs showed that 4-bit ADCs are sufficient to achieve near-optimal
performance in massive MIMO OFDM systems \cite{studer2016quantized}.
To facilitate the learning of likelihood probabilities in 1-bit ADC maximum likelihood data detector, a robust learning with artificial noise was presented in \cite{choi2019robust}.
The authors in \cite{zhang2016mixed, liang2016mixed} developed data detectors for mixed-ADC systems that assign either 1-bit or infinite-resolution depending on channel gain.
In addition, a resolution-adaptive ADC system with a bit-allocation algorithm was proposed while outperforming the conventional fixed-precision low-resolution ADC systems in terms of both spectral and energy efficiency \cite{choi2017resolution}.
For analytic tractability, the severe non-linearity of the low-resolution quantizer was linearized using Bussgang decomposition \cite{jacobsson2017throughput} and the additive quantization noise model (AQNM) \cite{orhan2015low, xu2019uplink, choi2019two, choi2020quantized}, while producing useful algorithms and insightful analytical results.

For downlink transmission, a number of works have reduced hardware cost using low-resolution DACs
\cite{jacobsson2017quantized,li2017downlink,park2021construction}.
Using Bussgang decomposition, the authors in \cite{jacobsson2017quantized} demonstrated a marginal communication gap in achievable rates for  linear precoders with 3 to 4-bit DACs compared to infinite-resolution DACs and further proposed a non-linear precoder in 1-bit DAC systems via relaxation and sphere precoding.
The use of $2.5\times$ more transmit antennas can compensate for the spectral efficiency loss brought on by the use of 1-bit DACs according to the analyzed rate of the quantized downlink systems with matched-filter precoding \cite{li2017downlink}.
Based on constructive interference and decision regions, low-complexity symbol-level precoding methods for 1-bit DAC systems were developed for quadrature-amplitude-modulation constellations in \cite{park2021construction}.
For quantized downlink OFDM systems, the authors in \cite{jacobsson2019linear} derived a lower bound on achievable sum-rate using linear precoding and oversampling DACs.
The authors in \cite{zhang2019mixed} considered massive MIMO relaying systems with mixed-DACs and mixed-ADCs at the relay and derived exact and closed-form expressions for the achievable rate which approach infinite-resolution performance using only 2-3 bits thanks to strong synergy with large-scale antenna arrays.
By applying AQNM to massive MIMO systems with millimeter wave channels, the authors in \cite{ribeiro2018energy} introduced fully-connected and partially-connected hybrid BF architectures under low-resolution DACs which are more energy-efficient than conventional digital-only precoders.
%

As modern cellular communication systems are primarily limited by interference, research efforts and later standards supported CoMP to coordinate transmission by multiple base stations (BSs) in order to reduce inter-cell interference, thereby improving data rates and coverage
\cite{jungnickel2014role,irmer2011coordinated,bengtsson1999optimal,ng2008distributed,shin2017coordinated,rashid1998joint,rashid1998transmit,song2007network,dahrouj2010coordinated}.
 The authors in \cite{jungnickel2014role} demonstrated that the synergy between CoMP and massive MIMO is advantageous to 5G communication systems due to a more robust link, localized interference, and reduced backhaul overhead.
The feasibility of CoMP was demonstrated for both UL and DL in physical testbeds with improved average throughput and cell edge throughput \cite{irmer2011coordinated}.

To improve data rate and satisfy demanding requirements of cellular systems, many papers have contributed to beamforming design \cite{bengtsson1999optimal,ng2008distributed,shin2017coordinated}.
The DL BF problem was cast as a semidefinite programming problem and efficiently solved via interior point methods \cite{bengtsson1999optimal}.
Noting that signals from neighboring BSs have significant impact, the authors in \cite{ng2008distributed} proposed a near-optimal distributed DL BF algorithm based on message passing between neighboring BSs without requiring centralized processing.
To improve the data rate of cell-edge users, the CoMP BF problem based on 
interference alignment was also studied in a non-orthogonal multiple access system \cite{shin2017coordinated}.

To further unleash the potential of CoMP systems, joint optimization of beamforming and power control has been proposed \cite{rashid1998joint,rashid1998transmit,song2007network,wiesel2005linear}.
For UL transmission, authors in \cite{rashid1998joint} proposed a fixed-point algorithm that jointly solves BF and PC problems and proves the existence of at least one optimal solution.
The authors in \cite{rashid1998transmit} formulated a virtual UL system model whose combiner and power control solutions are used to iteratively update the DL BF solution. 
For multiuser MIMO systems, linear programming UL–DL duality led to centralized and decentralized algorithms to efficiently solve the joint BF and PC problem \cite{song2007network}.
Authors in \cite{wiesel2005linear} derived an iterative algorithm that finds optimal UL PC and DL BF solutions based on Lagrangian theory and brought insights to \cite{dahrouj2010coordinated} which generalized the algorithm to multicell multiuser MIMO systems. 





Recently, a coarsely quantized CoMP BF and PC problem  was studied for OFDM systems in \cite{choi2020quantized}.
The UL-DL BF duality was extended to the quantized OFDM systems, and an iterative algorithm for solving the DL total transmit power minimization problems was developed by leveraging the duality \cite{choi2020quantized} without requiring explicit estimation of inter-cell interference. 
It was shown that the proposed algorithm can be performed in a distributed fashion by estimating the covariance matrix of the received signals at local BSs. 
The problem considered in \cite{choi2020quantized} only focused on the total power minimization with the quality-of-service constraints. 
Provided that power amplifiers escape the efficient amplification regime when RF input exceeds a certain value and one power amplifier is usually allocated for each transmit antenna, minimizing total power cannot guarantee the efficiency of hardware components.
Therefore, it is necessary to minimize the peak power consumed across all antennas to simplify the design of the power amplifier and avoid distortion from the nonlinearity of the power amplifier in the high-power regime.
It's worth mentioning that formulating the problem to minimize peak power, as opposed to the sum power objective, leads to a more intricate derivation and algorithm development.
The reduction of peak power has the potential to enhance the practicality and help us scale up the multi-cell and multi-user communication network.
Employing cell-free MIMO systems can be considered as a scalable way to implement CoMP \cite{interdonato2019ubiquitous}; however, cell-free MIMO solutions are missing either per-antenna level power restriction or the joint optimization of DL beamformer and UL power control \cite{zhang2018mixed,bjornson2020scalable}.
Recently, \cite{mai2020downlink} worked on the DL spectral efficiency (SE) max-min problem under sum power constraint of access points whereas \cite{wang2022uplink} investigated UL SE performance of various implementations for an arbitrary precoding and power control; however, allowing amplifiers to remain linear while reducing per-antenna power can improve the practicality of such studies.
Therefore, a thorough study on the quantized DL CoMP BF and UL power control for minimizing the transmit power with per-antenna power constraints would make a worthwhile contribution toward a more practical CoMP deployment.

\subsection{Contributions}
In this paper, we consider downlink multicell massive MIMO OFDM communication systems.
The BS are equipped with low-resolution data converters, i.e., DACs and ADCs, and cooperate for BF and PC. 
In such a system, we investigate a DL antenna power minimization problem with quality-of-service constraints.
The  contributions are summarized as follows:

\begin{itemize}
\item {\bf Formulating the DL antenna power minimax problem and deriving its dual.}
We aim to minimize the maximum transmit power over all transmit antennas and formulate the problem with individual signal-to-quantization-plus-interference-and-noise ratio (SQINR) constraints.
As the main contribution of this paper, we derive the Lagrangian dual of the primal DL OFDM problem, which can be considered as a virtual UL OFDM transmit power minimization problem with uncertain noise covariance matrices.
This finding extends the previous DL-UL duality under per-antenna power constraints \cite{yu2007transmitter} to the quantized OFDM systems.
By transforming the DL OFDM problem to a strictly feasible second-order cone program (SOCP), we show that  strong duality holds between the primal DL OFDM problem and its associated dual, i.e., the virtual UL OFDM problem in our work.
    
\item {\bf Developing an iterative minimax algorithm to provide a feasible solution}. 
Leveraging the strong duality, we develop an iterative algorithm to solve the primal DL OFDM BF problem;
we first solve the dual UL OFDM problem for a fixed set of UL noise covariance matrices.
 We then compute the DL beamformer via linear transformation of the obtained UL solution.
Using the DL beamformer, we update the UL noise covariance matrices via projected subgradient ascent method.
Finally, we repeat the steps until the UL noise covariance matrices converge.
Although the DL BF problem can be cast to a semidefinite programming, the proposed direct iterative updates are in general more efficient and insightful.

\item 
{\bf Validating the proposed algorithm through extensive simulations}.
Simulation results validate the derived results and algorithm in both wideband and narrowband scenarios.
The proposed algorithm outperforms approaches that do not impose per-antenna constraints in terms of the maximum antenna transmit power consumption.
We further show the significant advantages of the proposed algorithm with per-antenna constraints by comparing peak-to-average-power ratio (PAPR) which is a paramount measure for the design of the power amplifiers and other nonlinear electronics in OFDM communication systems.

\end{itemize}

The rest of this paper is organized as follows. 
In Section~\ref{sec:sys_model}, we present the wideband OFDM system under coarse quantization at the BSs and establish the maximum power minimization problem with target SQINR and per-antenna power constraints. 
Section \ref{sec:duality} proves duality with zero duality gap between the primal DL and dual UL problems.
Inspired by the strong duality, we derive the optimal DL solution for a given dual solution and noise covariance matrix in Section \ref{subsec:DL_sol} and propose an iterative algorithm to find the final solution in Section \ref{subsec:algorithm}.
In Section \ref{sec:simulation}, the proposed method is evaluated for various configurations along with other benchmarks with respect to the maximum transmit power and PAPR.
Section~\ref{sec:conclusion} concludes the paper.

{\it Notation}: $\bf{A}$ and $\bf{a}$ denote a matrix and a column vector, respectively. 
$\mathbf{A}^{H}$ and $\mathbf{A}^T$  denote conjugate transpose and transpose, respectively. 
$[{\bf A}]_{i,:}$ and $ \mathbf{a}_i$ indicate the $i$th row and column vectors of $\bf A$. 
We denote $a_{i,j}$ as the $\{i,j\}$th element of $\bf A$ and $a_{i}$ as the $i$th element of $\bf a$. 
$\mathcal{CN}(\mu, \sigma^2)$ is a complex Gaussian distribution with mean $\mu$ and variance $\sigma^2$. 
The diagonal matrix $\rm diag(\bf A)$ has $\{a_{i,i}\}$ as its diagonal entries, and $\rm diag (\bf a)$ or $\rm diag({\bf a}^T)$ creates a diagonal matrix with $\{a_i\}$ as its diagonal entries. 
A block diagonal matrix is denoted as ${\rm blkdiag}({\bf A}_1, \dots,{\bf A}_{N})$. 
${\bf I}_N$ is a $N\!\times\!N$ identity matrix.
${\bf 1}_N$ and ${\bf 0}_N$ are $N \!\times\! 1$ vector of ones and zeros, respectively.
We represent the vectorization operation of a matrix $\bA$ by $\rm vec({\bf A})$.
$\otimes$ denotes Kronecker product operator.
$\|\bf A\|$ represents $L_2$ norm.
$\preceq$ denotes matrix inequality. 
${\rm max}(a, b)$ represents an element-wise max function. 
$\bbE[\cdot]$ represents the expectation operator.

\section{System Model}
\label{sec:sys_model}
\begin{table}[!t]
\caption{A list of Key Expressions}
\begin{center}
\begin{tabular}{cc}
    \cline{1-2}
    \hlineB{2}
    {\bf Notation} & {\bf Description}\\
	\hlineB{2}
    $N_b$ & Number of BS antennas  \\
	\hline
 $N_c$& Number of cells
    \\
     \hline
     $N_u$& Number of users  \\
     \hline
     $K$ & Number of subcarriers  \\
    \hline 
    $\bx_i$& OFDM symbols at BS$_i$\\
    \hline
    $\alpha$ & Quantization gain\\
    \hline
    $\bq_i$& Quantization noise vector at BS$_i$ \\
    \hline
    ${\rm Q}_{i,u}(k)$ & Downlink quantization error \\
    \hline
    $\bg_{i,j,u}(k)$& Uplink channel from $u$th user in cell $j$ to BS$_i$ \\
    \hline
    $\bg_{j,i,u}^H(k)$& Downlink channel from BS$_j$ to $u$th user in cell $i$ \\
    \hline
    $\lambda_{i,u}(k)$ & Transmit power \\
    \hline
    $\bw_{i,u}(k)$ & Beamforming vector \\
    \hline
    $\bD_i$& Noise covariance matrix at BS$_i$  \\
    \hline
    $\bff_{i,u}(k)$ & Arbitrary uplink combiner \\
    \hline
    $\bff_{i,u}^{\sf MMSE}(k)$ & Uplink MMSE combiner \\
    \hline
    $\bZ_{i,u}(k)$ & Covariance of interference-plus-quantization-plus-noise term\\
 \hlineB{2}
\end{tabular}
\end{center}
\end{table}

\begin{figure}[!t]\centering
	\includegraphics[width=1\columnwidth]{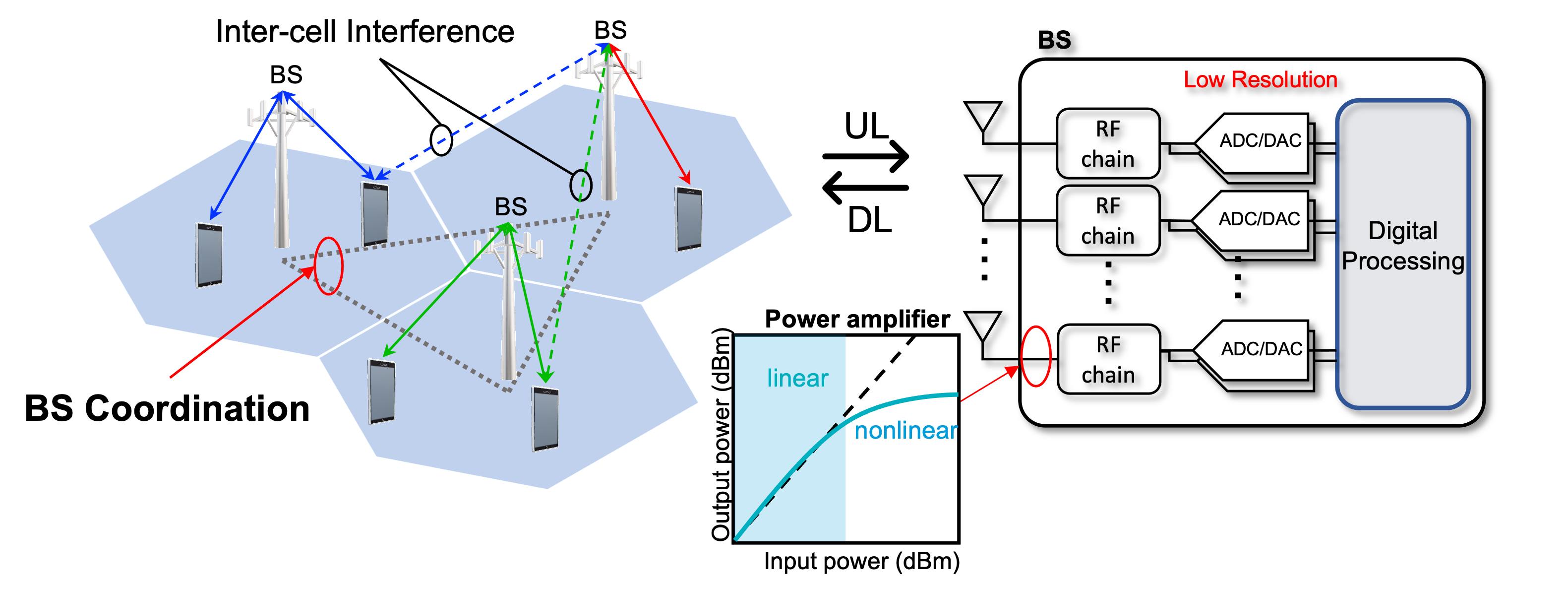}
	\caption{Multicell and multiuser MIMO system when each base station (BS) employs low-resolution ADCs and DACs. 
    Nonlinear distortions occur after a certain saturation point. } 
	\label{fig:system}
\end{figure}

\subsection{Network and Downlink Signal Model}
We consider a wideband multicell and multiuser MIMO network with $N_c$ cells, $N_u$ single-antenna users per cell, and $K$ subcarriers.
All BSs are equipped with $N_b$ antennas.
The BS in cell $i$ is denoted as BS$_i$ which precodes signals to serve $N_u$ users in cell $i$.
We assume the BSs for all $N_c$ cells cooperate and are equipped with low-resolution DACs with $b$ DAC quantization bits as shown in Fig.~\ref{fig:system}.
Time-division duplexing (TDD) is assumed to exploit channel reciprocity.
The BSs are considered to have perfect channel state information (CSI) for all channels as they cooperate.
In other words, the serving BS is assumed to have complete access to the wireless channels of interfering BSs to jointly tailor precoding vectors across all BSs.


We consider a  DL broadcast system where BS$_i$ generates $\bs_i(k)\!\sim\!\mathcal{CN}({\bf 0}_{N_u},{\bf I}_{N_u}) \in \bbC^{N_u}$ that contains the dedicated symbols for $N_u$ users in cell $i\in\{1,\ldots,N_c\}$ at subcarrier $k\in\{0,\ldots,K-1\}$.
We use $\bW_i(k)=\left[\bw_{i,1}(k),\ldots,\bw_{i,N_u}(k)\right]  \in \bbC^{N_b \times N_u} $ to denote the matrix of precoders, and the $u$th symbol of $\bs_i(k)$ is precoded by $\bw_{i,u}(k)$. 
We define the precoded frequency domain symbol for user $u$ in cell $i$ at subcarrier $k$ as $\bu_i(k) = \bW_i(k)\bs_i(k)$.
Let $\bx_i(k) \in \bbC^{N_b}$ be the vector of OFDM symbols of $N_u$ users in cell $i$ at time $k$.  
We stack the OFDM symbol vectors at BS$_i$ over $K$ time slots as $\bx_i = [\bx^T_i(0),\dots,\bx^T_i(K-1)]^T $, which is given as
\begin{align}
    \bx_i &= (\bW_{\rm DFT}^H\otimes \bI_{ N_b})\bu_i
    = {\pmb \Psi}_{N_b}^H\bW_i\bs_i\in \bbC^{KN_b},
\end{align} 
where $\bW_{\rm DFT} \in \bbC^{K\times K}$ represents a normalized discrete Fourier transform (DFT) matrix, $\bu_i=[\bu_i^T(0),\dots, \bu_i^T(K-1)]^T \in \bbC^{KN_b}$, ${\pmb \Psi}_{N_b} = \bW_{\rm DFT}\otimes \bI_{ N_b}$, $\bW_i \! =\! {\rm blkdiag}\big(\bW_i(0),\dots,\bW_i(K-1)\big) \!\in\! \bbC^{KN_b\times KN_u}$, and $\bs_i = [\bs_i^T(0),\dots,\bs_i^T(K-1)]^T \in \bbC^{KN_u}$.

Before transmitting the signals, $\bx_i$ is quantized by the low-resolution DACs with $b$ quantization bits.
We adopt the AQNM~\cite{choi2020quantized, orhan2015low} to elicit a linearized approximation of the quantization process derived from assuming a scalar minimum-mean-squared-error (MMSE) quantizer with Gaussian signaling.
Since the general trend of the capacity curve with low-resolution quantizers is well captured with the AQNM \cite{orhan2015low}, it is meaningful to use the AQNM for tractable analysis of quantized systems.
The quantized signal vector of input $\bx_i$ with the AQNM is represented as
\begin{align}
    Q(\bx_i) \approx  \bx_{{\rm q},i} &= \alpha \bx_i + \bq_i \\
    &= [\bx^T_{{\rm q},i}(0),\dots,\bx^T_{{\rm q},i}(K-1)]^T \in \bbC^{KN_b}
\end{align}
where $Q(\cdot)$ is a quantizer applied for real and imaginary parts,  $\bq_i =[\bq_i^T(0),\dots,\bq_i^T(K-1)]^T$ is the stacked quantization noise vector at BS$_i$ with $\bq_i(k)$ denoting the quantization noise vector for the quantization input of $\bx_i(k)$, $\bx_{{\rm q},i}(k)=\alpha\bx_i(k)+\bq_i(k)$ is the quantized signal transmitted by BS$_i$ at the $k$th time slot, and $\alpha$ is the quantization gain defined as $\alpha \!=\! 1 \!-\! \beta$ where the $\beta$ values are quantified in Table 1 in \cite{fan2015uplink}. 
We assume that the quantization noise is uncorrelated with $\bx_i$ and follows a complex Gaussian distribution, i.e., $\bq_i  \!\sim\! \cC\cN({\bf 0}, \bC_{{\bf q}_i})$ which is the worst case in terms of the achievable rate.
The covariance matrix $\bC_{{\bf q}_i}$ is computed as  \cite{choi2020quantized}
\begin{align}
    \label{eq:Cqq_dl_ofdm}
    \bC_{{\bf q}_i} = \alpha(1-\alpha){\rm diag}\big({\pmb \Psi}_{N_b}^H\bW_i\bW_i^H{\pmb \Psi}_{N_b}\big).
\end{align}

By the broadcast channel, each user collects signals from all BSs. 
Stacking over $K$ subcarriers after removing cyclic prefix (CP) and applying the DFT, the users in cell $i$ receive the signal of 
\begin{align}
    \label{eq:y_dl_ofdm}
    \by_i 
    =\alpha\bG_{i,i}^H\bW_i\bs_i + \alpha\sum_{j\neq i}^{N_c}\bG_{j,i}^H\bW_j\bs_j + \tilde{\bq}_j + \tilde{\bn}_i \in \bbC^{KN_u}.
\end{align}
The DL frequency domain channels combined accross all $K$ subcarriers are defined as $\bG^H_{j,i}= {\rm blkdiag}(\bG^H_{j,i}(0),\cdots,\bG^H_{j,i}(K-1)) \in \bbC^{KN_u \times KN_b}$, where $\bG_{j,i}(k) = \sum_{\ell = 0}^{L-1}\bH_{j,i,\ell} \,e^{-\frac{j2\pi k\ell}{K}}$ is the individual UL frequency domain  channel for the $k$th subcarrier  between BS$_j$ and users in cell $i$ (by channel reciprocity, which is an inherit feature of TDD systems, $\bG_{j,i}(k)$ is interpreted as the UL frequency domain channel matrix), $\bH_{j,i,\ell}$ is the UL time domain channel between BS$_j$ and users in cell $i$ for the $\ell$th channel tap, and $L$ is the delay spread. 
In addition, $\tilde{\bq}_j=[\tilde \bq_j^T(0),\dots, \tilde \bq_j^T(K-1)]^T  = \sum_{j=1}^{N_c}\bG_{j,i}^H{\pmb \Psi}_{N_b}\bq_j $ and $\tilde{\bn}_i  =[\tilde \bn_i(0), \dots, \tilde \bn_i(K-1)]^T= {\pmb \Psi}_{N_u}\bn_i$  where ${\pmb \Psi}_{N_u} = \bW_{\rm DFT}\otimes \bI_{N_u}$ and  $\bn_i = [\bn^T_i(0),\dots,\bn_i^T(K-1)]^T \sim \cC\cN({\bf 0}_{KN_u}, \sigma\bI_{KN_u}) $ denotes the stacked additive white Gaussian noise (AWGN) vector over $K$ time slots for $N_u$ users in cell $i$.
The received signal at user $u$ in cell $i$ for subcarrier $k$ is  expressed as
\begin{align}
   \label{eq:yiu_dl_ofdm}
    y_{i,u}(k) &=  \alpha\bg_{i,i,u}^H(k)\bw_{i,u}(k)s_{i,u}(k)
    \\  \nonumber
    &+ \!\alpha \!\!\! \sum_{(j,v) \neq (i,u)}^{N_c, N_u} \!\!\!\!\!\bg_{j,i,u}^H(k)\bw_{j,v}(k)s_{j,v}(k) + \tilde{q}_{i,u}(k) + \tilde{n}_{i,u}(k),
\end{align}
where  $\bg_{j,i,u}(k)$ indicates the $u$th column of $\bG_{j,i}$ and $s_{i,u}(k)$, $\tilde q_{i,u}(k)$, and $\tilde n_{i,u}(k)$ represent the $u$th element of $\bs_i(k)$, $\tilde \bq_i(k)$, and $\tilde \bn_i(k)$, respectively.

\subsection{Problem Formulation}
Based on \eqref{eq:Cqq_dl_ofdm}-\eqref{eq:yiu_dl_ofdm}, 
the DL SQINR for user $u$ in cell $i$ at subcarrier $k$ is represented as \footnote{Here, we assume perfect CSIT. To consider a channel estimation error for imperfect CSIT case, we may need to exploit channel error covariance matrices to reformulate the SQINR in terms of the covariance. Then similar analysis in our work can be applied by incorporating the impact of imperfect CSIT. We shall leave it for a future work.}
\begin{align}
    \label{eq:SQINR_dl_ofdm}
    \Gamma_{i,u}(k)=
    \frac{\alpha^2|\bg_{i,i,u}^H(k)\bw_{i,u}(k)|^2}{\alpha^2 \sum_{(j,v) \neq (i,u)}^{N_c, N_u}\!|\bg_{j,i,u}^H(k)\bw_{j,v}(k)|^2+ {\rm Q}_{i,u}(k) +\! \sigma^2},
\end{align}
whose first term in the denominator is interference and ${\rm Q}_{i,u}(k)$ is quantization error defined as
\begin{equation}
    {\rm Q}_{i,u}(k) = \sum\nolimits_{j=1}^{N_c}\ubg_{j,i,u}^H\!(k){\pmb \Psi}_{N_b}\bC_{{\bf q}_i}{\pmb \Psi}_{N_b}^H\ubg_{j,i,u}\!(k), \label{eq:quantizationnoise}
\end{equation}
where $\ubg_{j,i,u}\!(k)$ is the channel  for subcarrier $k$ of user $u$ extracted as the ($kN_u + u$)th column of $\bG_{j,i}$.
Since BS$_i$ sends $\bx_{{\rm q},i}(k)$ at the $k$th time slot, the average transmit power of the $m$th antenna of BS$_i$ at the $k$th time slot is written as $\Big[\bbE[\bx_{{\rm q},i}(k)\bx_{{\rm q},i}^H(k)]\Big]_{m,m}$.
Using \eqref{eq:SQINR_dl_ofdm}, the DL OFDM transmit power minimization problem with per-antenna power and SQINR constraints is formulated as 
\begin{align}
    \label{eq:dl_problem}
   \mathop{{\text{minimize}}}_{{\bf w}_{i,u}(k),\ p_{\rm 0}}& \;\;  p_{\rm 0}\\
   \label{eq:dl_SQINR}
   {\text{subject to}} & \;\; \Gamma_{i,u}(k) \geq \gamma_{i,u}(k) \quad \forall i,u,k  \\
   \label{eq:per_antenna_const}
   &\  \Big[\bbE[\bx_{{\rm q},i}(k)\bx_{{\rm q},i}^H(k)]\Big]_{m,m}\leq p_{\rm 0} \quad \forall i,k,m,
\end{align}
where $p_0$ denotes the peak transmit power across $N_bN_c$ BS antennas. 
Making $p_0$ smaller results in a more precise and economical usage of the power amplifiers by limiting the dynamic range.

We aim to identify the DL beamformer $\bw_{i,u}(k)$ for all $i,u,k$ that can minimize $p_0$ while satisfying all SQINR constraints.
As shown in \eqref{eq:SQINR_dl_ofdm}-\eqref{eq:quantizationnoise}, the fact that $\Gamma_{i,u}(k)$ contains the quantization noise as well as $\bw_{i,u}(k)$ for all $i,u,k$ makes the problem more challenging to solve.
Since a closed-form algebraic solution is not available, we seek an efficient algorithm to find a numerical solution.
Instead of solving the optimization problem directly, we first derive a Lagrangian dual and then solve the problem in the dual domain with an efficient solver.
We will show that the DL BF solution is ultimately calculated by exploiting its associated dual solution.

\section{Duality between Downlink and Uplink}
\label{sec:duality}

In this section, we first introduce the corresponding uplink problem and then show that the uplink and downlink problems have strong duality.  
The primary challenge in deriving a dual problem is the quantization noise terms coupled with both beamformers and OFDM modulation. 

\subsection{Dual UL OFDM Systems with Low-Resolution ADCs}
\label{sec:uplinkformulation}

Let $\bs^{\rm ul}_i(k) \in \bbC^{N_u}$ denote the vector of symbols from $N_u$ users in cell $i$ at subcarrier $k$ and let 
$\bu^{\rm ul}_i(k) = {\pmb \Lambda}_i(k)^{1/2}\bs^{\rm ul}_i(k)$
where ${\pmb \Lambda}_i(k) = {\rm diag}\big(\lambda_{i,1}(k),\dots,\lambda_{i,N_u}(k)\big)$ is the collection of non-negative transmit power.
With $\bx_i^{\rm ul}(k)$ denoting the vector of OFDM symbols of $N_u$ users in cell $i$ at time $k$,  
we stack the OFDM symbol vectors as 
\begin{align}
    {\bx}_i^{\rm ul} &= [\bx^{\rm ul}_i(0)^T,\dots,\bx^{\rm ul}_i(K-1)^T]^T\\
    &= (\bW_{\rm DFT}^H \otimes \bI_{N_u}){\bu}^{\rm ul}_i = {\pmb \Psi}_{N_u}^H{\pmb \Lambda}^{1/2}_i{\bs}^{\rm ul}_i,
\end{align}
where ${\pmb \Psi}_{N_u} \!=\! (\bW_{\rm DFT} \otimes \bI_{N_u})$, ${\bu}_i^{\rm ul} = [\bu^{\rm ul}_i(0)^T,\dots,\bu^{\rm ul}_i(K-1)^T]^T$, ${\bs}_i^{\rm ul} = [\bs^{\rm ul}_i(0)^T,\dots,\bs^{\rm ul}_i(K-1)^T]^T$, and ${\pmb \Lambda}_i\! =\! {\rm blkdiag}\big({\pmb \Lambda}_i(0),\dots,{\pmb \Lambda}_i(K\!-\!1)\big)$.

We assume that the noise vector received at BS$_i$, i.e., $\tilde{\bn}^{\rm ul}_i(k)$, is unknown yet tunable by diagonal noise covariance $\bD_i$, i.e. $\tilde{\bn}^{\rm ul}_i(k)\sim\cC\cN({{\bf 0}_{N_b}}, \bD_i)$. 
After performing CP removal and applying DFT operation, the frequency-domain received signal at subcarrier $k$ is then given as
\begin{align}
    \label{eq:yk_ul_ofdm} 
    \by^{\rm ul}_i(k) = &\alpha\bG_{i,i}(k){\pmb \Lambda}_i(k)\bs^{\rm ul}_i(k) \nonumber\\
    & \!\!\! + \alpha\sum_{j\neq i}^{N_c}\bG_{i,j}(k){\pmb \Lambda}_j(k)\bs^{\rm ul}_j(k) 
+\tilde{\bq}^{\rm ul}_i(k) + \alpha \tilde{\bn}^{\rm ul}_i(k)  .
\end{align}
The combined signal for user $u$ at subcarrier $k$ is now given as $\bff_{i,u}^H(k)\by^{\rm ul}_i(k)$ where $\bff_{i,u}(k)$ is an arbitrary equalizer of subcarrier $k$ for user $u$ in cell $i$.

Accordingly, when considering a virtual uplink transmit power minimization problem for \eqref{eq:yk_ul_ofdm}, we can define the uplink SQINR constraint  as 
\begin{equation}
    \label{eq:ul_const}
    \max_{\bff_{i,u}(k)}{\Gamma}^{\rm ul}_{i,u}(k) \geq \gamma_{i,u}(k) \;\;\forall i,u,k,
\end{equation}
where the UL SQINR for user $u$ in cell $i$ at subcarrier $k$ is computed as \eqref{eq:sinr_ul_ofdm} in the next page,
\begin{figure*}
\begin{align}
    \label{eq:sinr_ul_ofdm}
    {\Gamma}^{\rm ul}_{i,u}(k) \!=\! \frac{ \alpha^2\lambda_{i,u}(k)|\bff_{i,u}^H(k)\bg_{i,i,u}(k)|^2 }{ \alpha^2 \!\sum_{(j,v)\neq (i,u)}^{N_c,N_u}\!\lambda_{j,v}(k)|\bff_{i,u}^H(k)\bg_{i,j,v}(k)|^2  \!+\! \alpha^2 \bff_{i,u}^H(k) \bD_i \bff_{i,u}(k) \!+\! \bff_{i,u}^H(k)\bC_{\!\tilde{\bq}^{\rm ul}_i(k)}\bff_{i,u}(k) }.
\end{align}
\rule{\textwidth}{0.7pt}
\end{figure*}
${\pmb \Psi}_{N_b}(k) \!=\! \left([\bW_{{\rm DFT}}]_{k+1,:}\otimes \bI_{N_b}\right)$, $\bG_i \!=\! [\bG_{i,1},\dots, \bG_{i,N_c}]$, 
${\pmb \Lambda} \!=\! {\rm blkdiag}\big({\pmb \Lambda}_{1},\dots,{\pmb \Lambda}_{N_c}\big)$, and $\bC_{\tilde{\bq}^{\rm ul}_i(k)}$ is  the uplink quantization error covariance matrix expressed by the AQNM as
\begin{align}
    &\bC_{\tilde{\bq}^{\rm ul}_i(k)} = \\
    &\alpha(1-\alpha) {\pmb \Psi_{N_b}}(k){\rm diag}\big({\pmb \Psi^H_{N_b}}\bG_i{\pmb \Lambda} \bG_i^H{\pmb \Psi_{N_b}}\!+\!\bD_i\big){\pmb \Psi^H_{N_b}}(k). \nonumber
\end{align}
The UL SQINR can be rewritten as
\begin{equation}
    {\Gamma}^{\rm ul}_{i,u}(k)\! =\! \frac{ \alpha^2\lambda_{i,u}(k)|\bff_{i,u}^H(k)\bg_{i,i,u}(k)|^2 }{ \bff_{i,u}^H(k)\bZ_{i,u}(k)\bff_{i,u}(k)}
\end{equation}
    where
    \begin{align} 
    	&\bZ_{i,u}(k) = \alpha^2 \!\!\!\!\! \sum_{(j,v)\neq (i,u)} 
    	\!\!\!\!\!\lambda_{j,v}(k) {\bf g}_{i,j,v}(k){\bf g}_{i,j,v}^H(k) \\
     &+ \alpha {\bD_i} 
     \!+\!\alpha(1-\alpha) {\pmb \Psi_{N_b}}(k){\rm diag}\big({\pmb \Psi^H_{N_b}}\bG_i{\pmb \Lambda} \bG_i^H{\pmb \Psi_{N_b}}\big){\pmb \Psi^H_{N_b}}(k) \nonumber
    \end{align}
    is the auto-covariance matrix corresponding to the interference-plus-quantization-plus-noise term in \eqref{eq:yk_ul_ofdm}.

    We define the MMSE equalizer as
    \begin{align}
    	\label{eq:mmse}
    	\bff_{i,u}^{\sf MMSE}(k) =  &\bZ_{i,u}^{-1}(k){\bf g}_{i,i,u}(k).
    \end{align}
    
    We note that \eqref{eq:mmse} maximizes the uplink SQINR in the constraint \eqref{eq:ul_const}.
\subsection{Downlink-Uplink Duality}
In Theorem~\ref{thm:duality_ofdm}, we derive a dual problem of  \eqref{eq:dl_problem} which is equivalent to the virtual UL OFDM problem with unknown noise covariance matrices.
\begin{theorem}[Duality]
    \label{thm:duality_ofdm}
    The Lagrangian dual problem of the DL problem in \eqref{eq:dl_problem} is equivalent to 
    \begin{align}
        \label{eq:dual_problem}
        \max_{{\bf D}_i }\min_{\lambda_{i,u}(k)\geq0}& \;\;  \sum_{i,u,k}^{N_c,N_u,K}\lambda_{i,u}(k)\sigma^2\\
        \label{eq:dual_SQINR_const}
        {\text{ \rm subject to}} & \;\; {\Gamma}^{\rm ul}_{i,u}(k) \geq \gamma_{i,u}(k),\\
        &\;\; \bff_{i,u}(k) = \bff_{i,u}^{\sf MMSE}(k), \\
        \label{eq:D_const1}
        &\;\;  \bD_i \succeq 0, \ \bD_i\in \bbR^{N_b\times N_b}: \text{\rm diagonal}, \\
        \label{eq:D_const2}
        &\;\;  {\rm tr}(\bD_i) \leq N_b \quad \forall i,u,k.
   \end{align}
We can consider the  problem above as a virtual UL transmit power minimization problem for the virtual UL system introduced in Section~\ref{sec:uplinkformulation} where BS$_i$ operates with unknown noise covariance matrix ${\bf D}_i$.  
    
\begin{proof}
    See Appendix~\ref{appx:duality_ofdm}.
\end{proof}
\end{theorem}
We remark that the Lagrangian dual problem in \eqref{eq:dual_problem} is considered to be an antenna power minimax problem with noise covariance constraints for the virtual UL ODFM system with low-resolution ADCs at the BSs shown in Theorem~\ref{thm:duality_ofdm}.

\begin{corollary}[Strong Duality]
\label{cor:strongduality}
 There exists zero duality gap between the DL problem formulation and its associated dual problem.
\begin{proof}
    See Appendix~\ref{appx:strongduality}.
\end{proof}
\end{corollary}


\section{Proposed Solution for Joint Beamforming}

In this section, we present the algorithm that can efficiently find the DL BF solution by leveraging the strong dual problem derived in Section \ref{sec:duality}.
In particular, we separate the dual problem in Theorem~\ref{thm:duality_ofdm} into the inner minimization and outer maximization problems and solve the problems in an alternating manner;
upon obtaining the dual solution, we use the dual solution to identify the primal solution and repeat through proper update and projection onto feasible sets.

\subsection{Optimal Downlink Precoder}
\label{subsec:DL_sol}

We first present the linear relationship between the optimal DL precoder and the dual solution, i.e., the virtual UL MMSE combiner in Corollary~\ref{cor:dl_precoder_ofdm}.
\begin{corollary}[Optimal DL Beamformer]
	\label{cor:dl_precoder_ofdm}
    With coefficients $\tau_{i,u}(k)$'s, an optimal DL beamformer can be obtained by establishing a linear transformation of the UL MMSE receiver, i.e.,  $\bw_{i,u}(k) = \sqrt{\tau_{i,u}(k)}\bff_{i,u}^{\sf MMSE}(k) \;\forall i, u, k$.
    Here, $\tau_{i,u}(k)$'s are derived from solving ${\btau} = {\bSigma}^{-1}{\bf 1}_{N_uN_cK}$, where   ${\btau}=[{\btau}^T(0), \cdots, {\btau}^T(K-1)]^T$ with ${\btau}(k) = [{\btau}_{1}^T(k), \cdots, {\btau}_{N_c}^T(k) ]^T$ and  ${\btau}_{i}^T(k) = [\tau_{i,1}(k), \cdots, \tau_{i,N_u}(k)]^T$, and ${\bSigma} = {\rm blkdiag}\big({\bSigma}(0),\dots,{\bSigma}(K-1)\big)$ whose submatrix is defined as
	\begin{equation}
		\label{eq:wb_constraint_matrix)}
  		{\bSigma}(k) = 
		\begin{pmatrix}
			{\bSigma}_{1,1}(k)  & \cdots & {\bSigma}_{1,N_c}(k) \\
			\vdots  & \ddots  & \vdots  \\
			{\bSigma}_{N_c,1}(k) & \cdots & {\bSigma}_{N_c,N_c}(k)
		\end{pmatrix},
	\end{equation}
	and
	\begin{align} 
		&[{\bSigma}_{i,j}(k)]_{u,v} = \\
            &\begin{cases}
			\frac{\alpha^2}{\gamma_{i,u}(k)}|\bg_{i,i,u}^H(k)\bff_{i,u}^{\sf MMSE}(k)|^2 \! - \alpha(1-\alpha)\!\sum_{\ell}\bff_{i,u}^{\sf MMSE}(\ell)^H{\pmb \Psi}_{N_b}(\ell)
			\\
			\quad \times {\rm diag}\left({\pmb \Psi}_{N_b}^H\bg_{i,i,u}(k)\bg_{i,i,u}^H(k){\pmb \Psi}_{N_b}\right){\pmb \Psi}_{N_b}^H(\ell)\bff_{i,u}^{\sf MMSE}(\ell), 
			\\
                \qquad
			\text{if } i=j \text{, } u=v, \\ 
			\nonumber
			- \alpha^2 |\bg_{j,i,u}^H(k)\bff_{j,v}^{\sf MMSE}(k)|^2\! - \alpha(1-\alpha)\!\sum_{\ell}\bff_{j,v}^{\sf MMSE}(\ell)^H{\pmb \Psi}_{N_b}(\ell) 
			\\
			\quad \times {\rm diag}\left({\pmb \Psi}_{N_b}^H\bg_{j,i,u}(k)\bg_{j,i,u}^H(k){\pmb \Psi}_{N_b}\right){\pmb \Psi}_{N_b}^H(\ell)\bff_{j,v}^{\sf MMSE}(\ell),
			\\
                \qquad
			\text{otherwise.}
		\end{cases}
	\end{align}
	
	\begin{proof}
	    See Appendix~\ref{appx:dl_precoder_ofdm}
\end{proof}
\end{corollary}

By Corollary~\ref{cor:dl_precoder_ofdm}, the optimal DL precoder $\bw_{i,u}(k)$ can be derived by using the virtual UL MMSE combiner $\bff_{i,u}^{\sf MMSE}(k)$ with a scalar weight $\sqrt{\tau_{i,u}(k)}$. 
The virtual UL MMSE equalizer, however, is a function of the uncertain noise covariance matrix $\bD_i$ and UL transmit power $\lambda_{i,u}(k)$ as shown in \eqref{eq:mmse}.
Accordingly, we further need to find the optimal virtual noise covariance matrix and UL transmit power to update its corresponding UL MMSE equalizer and DL precoder.

\subsection{Iterative Algorithm via Dual Uplink Solution}
\label{subsec:algorithm}

In this subsection, we characterize the virtual UL solutions by exploiting the strong duality and further adopt the fixed-point iteration with a projected subgradient ascent \cite{dahrouj2010coordinated} to maximize the objective function with respect to the noise covariance matrix of the virtual UL problem.

Instead of directly solving the dual problem in \eqref{eq:dual_problem}, we can decouple it into outer maximization and inner minimization problems where the latter is then written as
\begin{align}
        \label{eq:subproblem}
        f\!\left(\bD_i\right)=&\min_{\lambda_{i,u}(k)\geq0} \;\;  \sum\nolimits_{i,u,k}\lambda_{i,u}(k)\sigma^2\\
        \nonumber
        &{\text{ \rm subject to}}  \;\; \max_{\bff_{i,u}(k)}{\Gamma}^{\rm ul}_{i,u}(k) \geq \gamma_{i,u}(k) \;\;\forall i,u,k.
\end{align}
We notice that the above problem is interpreted as the inner optimization on $\lambda_{i,u}(k)$ of  the dual objective function $g(\tilde\bD_i, \lambda_{i,u}(k))$ where MMSE equalizer in \eqref{eq:mmse} is used for computing $\bff_{i,u}(k)$.
The virtual UL power control solution of \eqref{eq:subproblem} derived in Corollary~\ref{cor:solution_ofdm} is the solution of \eqref{eq:subproblem}. 
\begin{corollary}
\label{cor:solution_ofdm}
For given $\bD_i$, the optimal power for the virtual UL problem in \eqref{eq:dual_problem} is derived as
\begin{align}
    \label{eq:solution_ofdm}
    \lambda_{i,u}(k) = \frac{1}{\alpha \Big(1+1/\gamma_{i,u}(k)\Big){\bf g}_{i,i,u}^H(k) \bK_{i,k}^{-1}({\bLambda}) {\bf g}_{i,i,u}(k)},
\end{align}
where $\bK_{i,k}({\pmb\Lambda})$ is defined as
\begin{align}
        &{\bf K}_{i,k}({\pmb \Lambda}) \!=\! {\bf D}_{i} \! + \!\alpha \!\sum_{j,v} \!\lambda_{j,v}(k){\bf g}_{i,j,v}(k){\bf g}_{i,j,v}^H(k) \nonumber\\
        &+ (1\!-\!\alpha){\pmb \Psi}_{\!N_b}\!(k){\rm diag}\Big(\!{\pmb \Psi}_{N_b}^H\bG_i {\pmb \Lambda} \bG_i^H{\pmb \Psi}_{\!N_b}\!\Big){\pmb \Psi}_{\!N_b}^H\!(k).
        \label{eq:K_matrix}
    \end{align}
\end{corollary}
\begin{proof}
See Appendix~\ref{appx:solution_ofdm}.
\end{proof}
Since $\bK_{i,k}({\bf \Lambda})$ is a function of $\lambda_{i,u}(k)$, we iteratively update  $\lambda_{i,u}(k)$ until convergence. 
Once we obtain the optimal $\lambda_{i,u}(k)$ for the fixed $\bD_i$, we update $\bD_i$ by solving  the outer maximization for a fixed $\lambda_{i,u}(k)$ and then alternate the inner and outer problems until the solutions converge.
We note that the function $f\!\left(\bD_i\right)$ is concave in $\bD_i$ by the fact that  $f\!\left(\bD_i\right)$ is the objective function of a dual problem.
Based on this observation, we adopt a projected subgradient ascent method to maximize the objective function while satisfying the constraints on $\bD_i$ in \eqref{eq:D_const1} and \eqref{eq:D_const2}.

\begin{corollary}[Subgradient]
\label{cor:subgradient}
	 One of the subgradients of \eqref{eq:subproblem} in updating $\bD_i$ is obtained as 
	\begin{align}
	    \label{eq:subgradient}
	    {\rm diag}\Big(\sum\nolimits_{u,k}\bw_{i,k}(k)\bw_{i,k}^H(k)\Big).
	\end{align}
    \begin{proof}
	See Appendix~\ref{appx:subgradient}.
    \end{proof}
\end{corollary}

\begin{algorithm}[t]
	\caption{Quantized CoMP beamforming with per-antenna power (Q-CoMP-PA)
 \label{alg:Q_CoMP_AP}
	}
		{\bf Initialize} $\lambda_{i,u}^{(0)}(k)$, $\forall i,u,k$ and $\bD_{i}^{(0)}$, $\forall i$.\\
		\While{$|\bD_{i}^{(n)}-\bD_{i}^{(n-1)}|\geq \epsilon_{\bD}, \; \forall i$}{
		\While{$|\lambda^{(t)}_{i,u}(k)-\lambda^{(t-1)}_{i,u}(k)|\geq \epsilon_{\lambda}, \; \forall i,u,k$}{
		Compute $\bK_{i,k}({\pmb\Lambda}^{(t)})$ according to \eqref{eq:K_matrix} using $\lambda^{(t)}_{i,u}(k)$ and $\bD_i^{(n)}$.
		\\
        Update $\lambda_{i,u}^{(t+1)}(k)$ $\forall i, u, k$ according to  \eqref{eq:solution_ofdm} as
		\begin{align}
			\nonumber
			&\lambda^{(t+1)}_{i,u}(k) \gets
                \\\nonumber
                &\frac{1}{\alpha \Big(1+\frac{1}{\gamma_{i,u}(k)}\Big){\bf g}_{i,i,u}^H(k) \bK^{-1}_{i,k}({\pmb\Lambda}^{(t)}) {\bf g}_{i,i,u}(k)}
		\end{align}
		$t \gets t + 1$.
        }
		Find the UL MMSE equalizer $\bff_{i,u}^{\sf MMSE}(k)$ in \eqref{eq:mmse} with $\bD_i^{(n)}$ and  $\lambda_{i,u}^{(t)}(k)$.\\
		Compute the DL precoder $\bw_{i,u}(k)$ from Corollary~\ref{cor:dl_precoder_ofdm}.\\
		Update $\bD_i^{(n+1)}$ $\forall i$ using a subgradient ascent method as
		\begin{align}
		    \nonumber
		    \bD_{i}^{(n+1)} \!\gets\!  \bD_{i}^{(n)} \!+\! \eta\; {\rm diag}\Big(\sum\nolimits_{u,k}\bw_{i,u}(k)\bw_{i,u}^H(k)\Big).
		\end{align}
		Project $\bD_{i}^{(n+1)}$ onto the feasible set \eqref{eq:D_const1}-\eqref{eq:D_const2} $\forall i$ until converges as
		\begin{align}
	    	\nonumber
		    \bD_{i}^{(n+1)} \gets {\rm max}\left(0,\frac{\left({\rm tr}\left(\bD_{i}^{(n+1)}\right)-N_b\right)}{\left\|{\bf 1}_{N_b}\right\|^2}{\rm diag}({\bf 1}_{N_b})\right)
		\end{align}
		$n \gets n+1.$
		}
\Return{\ }{$\bw_{i,u}(k), \ \forall i, u, k$}.
\end{algorithm}
We can finally assemble the inner minimization and outer maximization problems into the complete DL problem in \eqref{eq:dual_problem}.
Recall that \eqref{eq:dual_problem} eventually becomes the minimax problem whose main objective is to maximize the concave $f(\bD_i)$ by updating $\bD_i$ while satisfying the constraints on $\bD_i$ and $f(\bD_i)$ wants to minimize the total transmit power of the virtual UL problem. 
Upon obtaining a solution of $f(\bD_i)$ for fixed $\bD_i$ as shown in Corollary~\ref{cor:solution_ofdm}, we update $\bD_i$ by taking a step in the direction of a positive subgradient to establish an increment in the objective function. 

Let $\bD_i^{(n)}$ denote $\bD_i$ at the $n$th iteration. 
Then, the update of $\bD_i^{(n)}$ is performed as
\begin{align}
    \bD_{i}^{(n+1)} =  \bD_{i}^{(n)} + \eta\; {\rm diag}\Big(\sum\nolimits_{u,k}\bw_{i,u}(k)\bw_{i,u}^H(k)\Big),
\end{align}
where $\eta>0$ is a step size.
We remark that the subgradient ascent method is guaranteed to converge to an optimal point since the UL problem is convex.
Since $\bD_i^{(n+1)}$ may escape the feasible domain, we project the updated $\bD_i$ onto the feasible set on $\bD_i$ in \eqref{eq:D_const1} and \eqref{eq:D_const2} as 
\begin{align}
    \label{eq:D_update1}
    \bD_{i}^{(n+1)} =  \frac{\left({\rm tr}\left(\bD_{i}^{(n+1)}\right)-N_b\right)}{\left\|{\bf 1}_{N_b}\right\|^2}{\rm diag}({\bf 1}_{N_b}),
\end{align}
and
\begin{align}
    \label{eq:D_update2}
    \bD_{i}^{(n+1)} =  {\rm max}\left(0,\bD_{i}^{(n+1)}\right), 
\end{align}
respectively. 
The proposed algorithm for the quantized CoMP BF with per-antenna power constraints (Q-CoMP-PA) is described in Algorithm~\ref{alg:Q_CoMP_AP}.

\begin{remark}[Narrowband Systems]
\label{rm:narrow}
\normalfont
The considered quantized OFDM system   naturally reduces to a quantized MIMO narrowband system with $K = 1$ subcarrier.
Accordingly, the derived duality, optimal solutions, and algorithm still hold for the quantized MIMO narrowband communications.
In this case, we omit the subcarrier index $k$ from the system parameters and variables.
\end{remark}

\subsection{Convergence of Sub-problems}

The convergence of each sub-problem can be guaranteed as follows:
regarding the inner minimization problem, the fixed-point iteration for given $\bD_i^{(n)}$ converges as shown in Corollary~\ref{cor:convergence}.
\begin{corollary}
\label{cor:convergence}
For any arbitrary initial points $\lambda_{i,u}^{(0)}(k)$, $\forall i,u,k$, the proposed fixed-point algorithm converges to a unique fixed point at which objective function of inner problem is minimized.


\begin{proof}
See Appendix~\ref{appx:convergence}.
\end{proof}

\end{corollary}
Therefore, the fixed-point iteration always converges to a unique fixed point that is the optimal solution of the inner optimization subproblem for a fixed covariance matrix.

Regarding the outer maximization problem, the iterative projection of $\bD_i^{(n)}$ converges for given $\bw_{i,u}(k)$  as discussed in \cite{bauschke1996projection}: the alternating projection on \eqref{eq:D_const1} and \eqref{eq:D_const2} converges to the intersection of the convex sets if there exists any point in the intersection.
Since $\bD_i=\bI_{N_b}$ belongs to the intersection, the intersection is not an empty set, thereby guaranteeing the convergence.

Accordingly, we theoretically show that each sub-problem can be solved with guaranteed convergence. 
For the convergence of the entire algorithm, we provide the numerical validation in Section~\ref{subsec:convergence}.

\subsection{Computational Complexity}
\label{sec:complexity}
The computational complexity of the algorithms is governed by the computation of $\lambda_{i,u}(k)$, which is dominated by the matrix inversion of $\bK_{i,k}({\bLambda})$. 
Exploiting the Hermitian symmetric and positive semi-definite properties of $\bK_{i,k}({\bLambda})$ can lead to a reduction in the number of floating point operations required for matrix inversion.
Consequently, the computational complexity of the per-cell and Q-CoMP algorithms are $O(T_1 K N_cN_uN_b^3)$ and $O(T_2 K N_cN_uN_b^3)$, respectively, where $T_1$ and $T_2$ denote the number of iterations required for obtaining a converged $\lambda_{i,u}(k)$.
Likewise, the inner optimization loop of the proposed Q-CoMP-PA with $T_3$ iterations has computational complexity of $O(T_3 K N_cN_uN_b^3)$. 
Since the Q-CoMP-PA requires an outer maximization loop with $T_0$ iterations, the total computational complexity of the proposed algorithm is $O(T_0 T_3 K N_cN_uN_b^3)$.
It is worth mentioning that the proposed method can scale linearly with respect to both the number of users and subcarriers.
We note that recently proposed coordinated beamforming algorithms consume a similar or more computational complexity.
For example,  \cite{chen2021joint} and \cite{jafri2022robust} respectively require $O(N_uN_b^4+N_c^6N_u^3)$ and $O(N_b^{3.5}N_u^{3.5})$ computations of multiplications and additions for each BS while both solutions are limited to the single subcarrier systems, i.e., $K=1$.

\section{Simulation Results}
\label{sec:simulation}


We evaluate the derived results of the proposed Q-CoMP-PA algorithm.
We also simulate the quantization-aware CoMP algorithm (Q-CoMP) presented in \cite{choi2020quantized} for comparison. 
Both methods utilize multicell coordination in designing DL precoders.
However, Q-CoMP-PA  aims to minimize the maximum antenna power whereas Q-CoMP minimizes the total transmit power without considering the per-antenna power constraints.
We compare the methods in terms of maximum transmit antenna  power, dynamic range, and PAPR for both wideband and narrowband systems.
We assume the target SQINR $\gamma$ is equal for all users and subcarriers.

\subsection{Wideband OFDM Communications}
\label{subsec:sim_wideband}
We consider the wideband OFDM systems assuming the delay spread is $L=3$ and the small scale fading follows Rayleigh fading with zero mean and unit variance.
We consider a $24 \ \rm GHz$ carrier frequency with $100\  \rm MHz$ bandwidth.
The adjacent BSs are $200\  \rm m$ apart and the minimum distance between any BS and user is $50 \, \rm m$.
For large scale fading, we adopt the log-distance pathloss model in \cite{akdeniz2014millimeter} with the $72\ \rm dB$ intercept,  the pathloss exponent of $2.92$, and the shadow fading whose shadowing variance is $8.7\ \rm dB$. Noise power is computed with $-174\ \rm dBm/Hz$ power spectral density and $5\ \rm dB$ noise figure. We also use the sector antenna gain of $15\ \rm dB$.

\begin{figure}[!t]\centering
	\includegraphics[width=1\columnwidth]{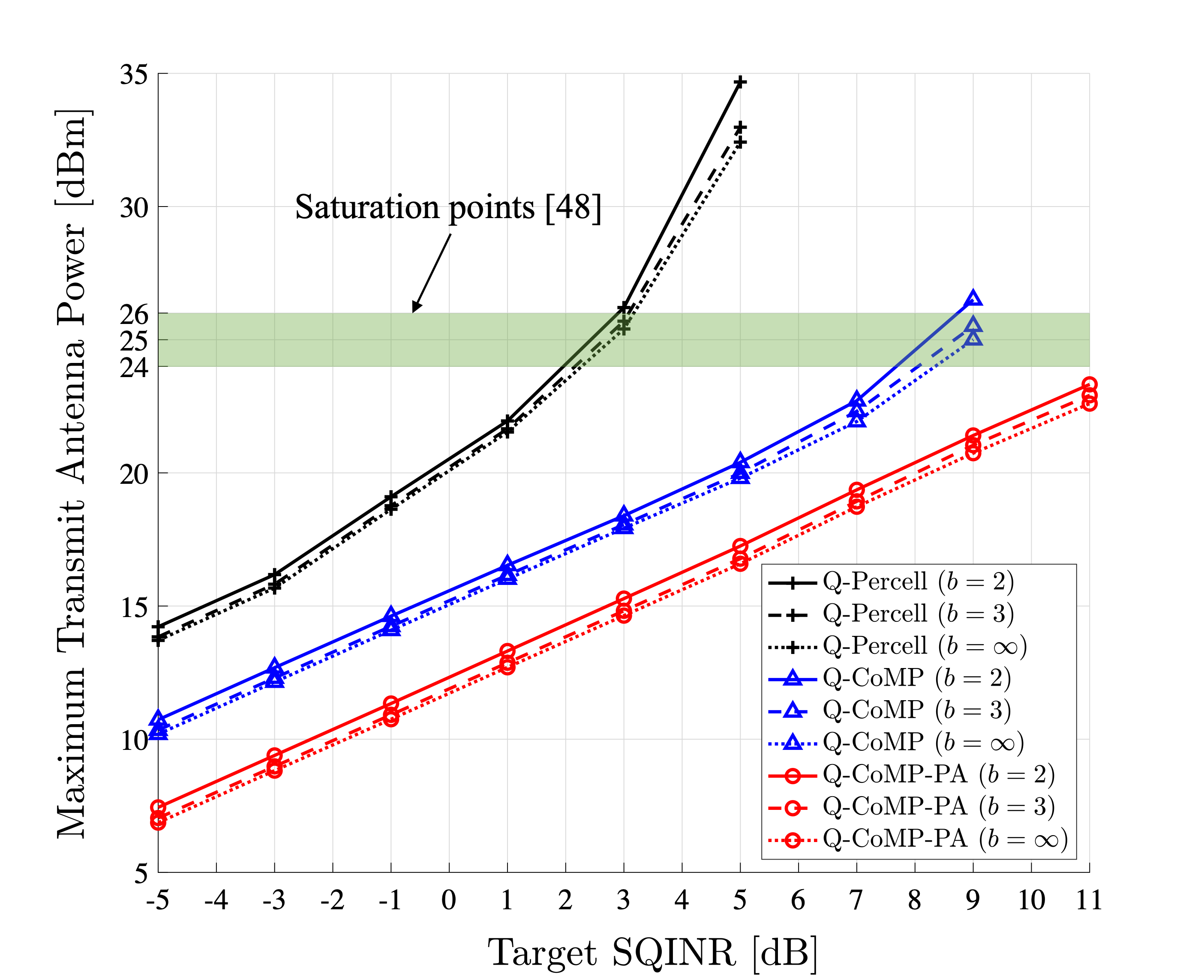}
	\caption{Maximum transmit power versus the target SQINR for the network with $N_b = 16$ antennas, $N_c = 4$ cells,  $N_u = 2$ users per cell, $K=32$ subcarriers, and $b \in \{2,3,\infty\}$ bits. 
} 
	\label{fig:power_wb}
\end{figure}

We  evaluate the performance of the Q-CoMP-PA and Q-CoMP algorithms in wideband OFDM communication systems. 
We also simulate the quantization-aware per-cell based CoMP (Q-Percell) algorithm as an additional benchmark by adapting the algorithm in \cite{rashid1998transmit} to the considered system.
For the Q-Percell algorithm, each BS first discovers its optimal solution by treating the inter-cell interference as noise and assuming that it is fixed.
Once the BSs derive solutions for the given inter-cell interference, the BSs share the converged solutions to update the interference power and compute solutions again.
These steps repeat until the solutions converge.
Consequently, Q-Percell is expected to be far sub-optimal than the other algorithms.

\begin{figure}[!t]\centering
	\includegraphics[width=1\columnwidth]{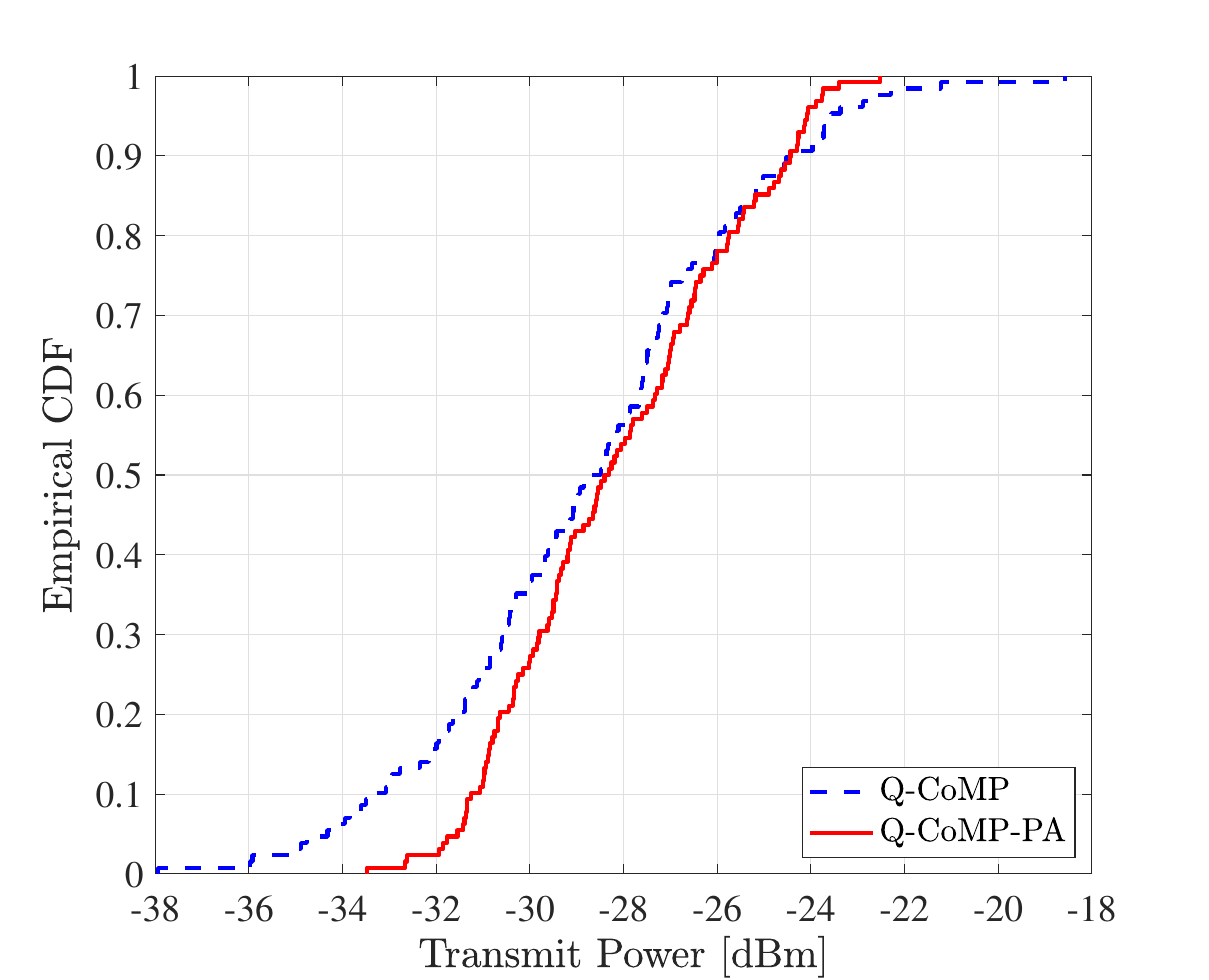}
	\caption{
		Empirical CDFs of transmit power values in the network with $N_b = 32$ antennas, $N_c = 4$ cells,  $N_u = 2$ users per cell, $K=64$ subcarriers, $b=3$ bits, and $\gamma=-1 \ \rm dB$ target.
        } 
	\label{fig:cdf}
\end{figure}

In Fig.~\ref{fig:power_wb}, we present the maximum transmit antenna power in DL direction across all $N_c N_b$ transmit antennas which is $p_0$ of the primal DL problem that we want to minimize for given target SQINRs.
We consider $N_b=16$ antennas, $N_c = 3$ cells, $N_u=2$ users per cell, $K=32$ subcarriers, and $b\!\in\!\{2,3,\infty\}$ DAC bits.
We see that the Q-Percell method consumes much higher transmit power than the Q-CoMP method and proposed Q-CoMP-PA method and shows divergence in the maximum transmit power as the target SQINR increases.
In contrast to the Q-Percell algorithm, the Q-CoMP and Q-CoMP-PA methods with multicell coordination scale linearly with the target SQINR without implausible power consumption or unfavorable divergence.
Nevertheless, based on the primal objective of the Q-CoMP-PA, the proposed algorithm can limit the maximum transmit power providing around $3 \ \rm dB$ reduction over the regular Q-CoMP. 
From both Q-CoMP and Q-CoMP-PA methods, when using infinite-resolution quantizers, we have lower peak power compared to the one with 2-bit or 3-bit data converters; however, the gap between $b=2$, $b=3$, and $b=\infty$ cases is marginal for both Q-CoMP and Q-CoMP-PA by properly taking the coarse quantization error and inter-cell interference into account. 

We further note that transmit power higher than the saturation input triggers RF nonlinearity, thereby causing a substantial reduction in RF output power and many undesirable additional frequencies compared to the ideal linear amplification regime of RF power amplifiers. 
As most of the state-of-the-art RF power amplifiers introduced in \cite{vasjanov2018review} have a saturation input of between $24 \ \rm dBm$ and $26 \ \rm dBm$, we can interpret from Fig.~\ref{fig:power_wb} that the per-cell and Q-CoMP methods escape the efficient amplification regime at around $3 \ \rm dB$ and $8 \ \rm dB$ target SQINR, respectively, while the proposed Q-CoMP-PA has stable and linear amplification at the simulated target SQINR values.



Fig.~\ref{fig:cdf} shows the empirical cumulative density function (CDF) of the transmit power by collecting the computed transmit power from all $N_bN_c$ BS antennas.
We employ a communication network with $N_b=32$, $N_c = 4$, $N_u=2$, $K=64$, $b=3$, and $\gamma=-1 \ \rm dB$.
Regarding the peak power, the proposed Q-CoMP-PA method achieves around $4 \ \rm dB$ lower maximum power over Q-CoMP, which corresponds to the main purpose of the primal problem.
In addition, the dynamic range defined as the gap between the minimum and maximum antenna power is narrower for Q-CoMP-PA compared to Q-CoMP, thereby showing more even power distribution across antennas.
Solving the maximum power minimization, as opposed to the sum power objective, can eventually simplify the hardware operation and avoid distortion from the nonlinearity of the power amplifier.
These results unleash the potential to enhance scalablity and practicality of CoMP-based MIMO systems.


\subsection{Narrowband Communications}

As a special case, we simulate over the narrowband channel, i.e.\ $K=1$, as stated in Remark~\ref{rm:narrow}.
We  assume that the small scale fading of each channel fading coefficient also follows Rayleigh fading with zero mean and unit variance.
For the large scale fading, we adopt the log-distance pathloss model in \cite{erceg1999empirically}. 
The distance between adjacent BSs is $2 \, \rm km$.
The minimum distance between BS and user is $100 \, \rm m$.
We consider a $2.4 \, \rm GHz$ carrier frequency with $10 \, \rm MHz$ bandwidth and use the same large scale fading model as the one used in Section~\ref{subsec:sim_wideband}


\begin{figure}%
    \centering
    \subfigure[Maximum transmit antenna power versus target SQINR with $N_c = 4$ cells, and $b \in \{2,3,\infty\}$.]
    {{\includegraphics[width=8cm]{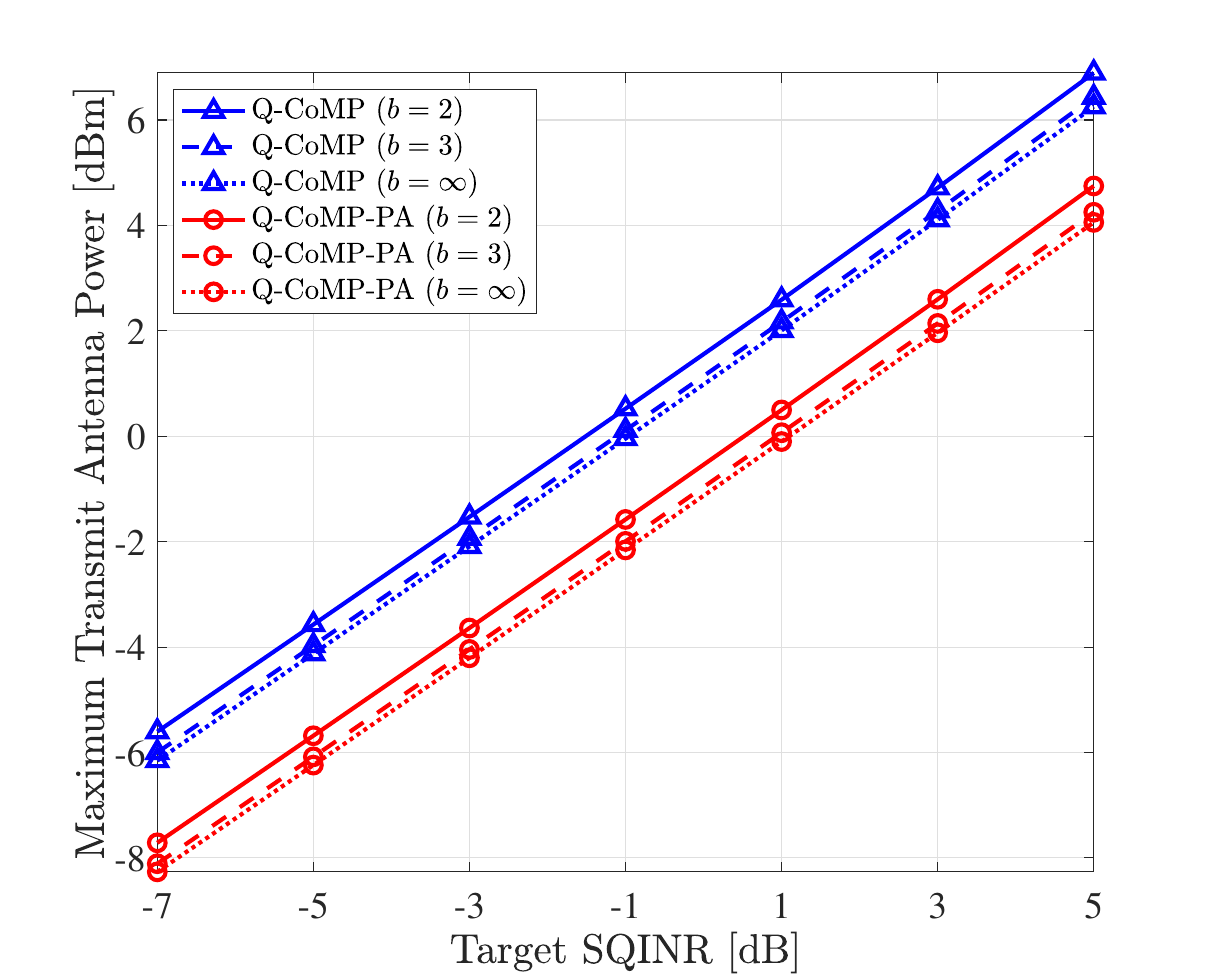} }}%
    \quad
    \subfigure[Empirical CDFs with $N_c = 5$ cells, $b=3$ quantization bits, $\gamma=2 \ \rm dB$ target SQINR.]{{\includegraphics[width=8cm]{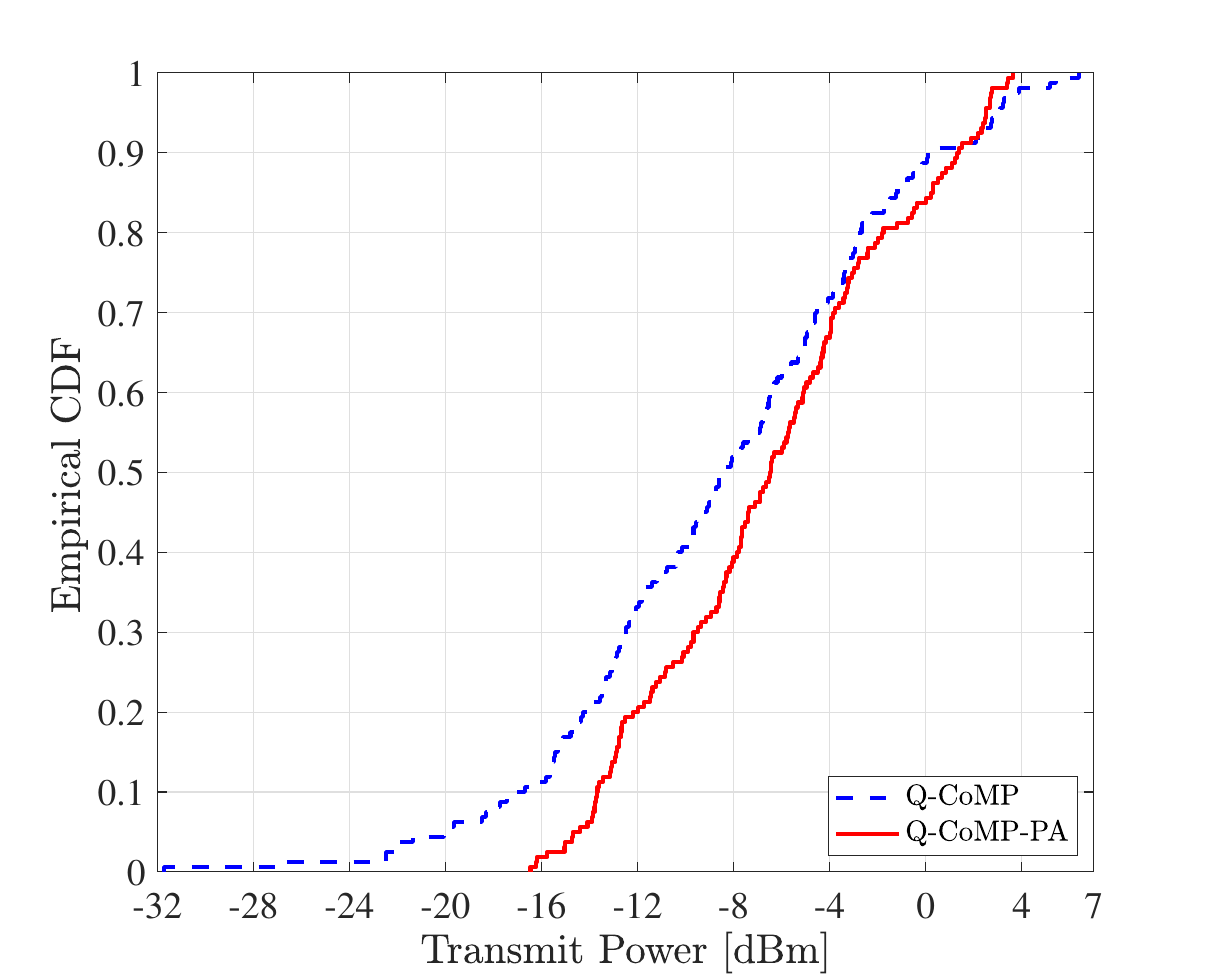} }}%
    \caption{Simulations for the narrowband communication networks with $N_b = 32$ antennas and $N_u = 2$ users per cell.}%
    \label{fig:narrowband}%
\end{figure}
Fig.~\ref{fig:narrowband}(a) illustrates the maximum transmit antenna power across $N_c N_b$ transmit antennas. 
We consider $N_b=32$, $N_c = 4$,  $N_u=2$, and  $b\in\{2,3,\infty\}$ bits.
The benefit of the proposed Q-CoMP-PA method identified in the wideband simulation can also be confirmed in the narrowband case.
By properly incorporating the quantization noise and inter-cell interference in designing $\bw_{i,u}$, the proposed Q-CoMP-PA method can achieve the $3 \ \rm dB$ reduction over Q-CoMP in peak transmit power without experiencing undesirable divergence at the simulated target SQINR values.

Fig.~\ref{fig:narrowband}(b) shows the CDF of the transmit power of all antennas.
We use $N_b=32$, $N_c = 5$, $N_u=2$, and $b=3$.
Q-CoMP-PA achieves more than $3 \ \rm dB$ reduction in the maximum transmit antenna power and operates with a much narrower dynamic range than  Q-CoMP.
This result validates that the proposed algorithm also outperforms  Q-CoMP in minimizing the transmit antenna power under the quality-of-service constraints for the narrowband communication system.


\subsection{Convergence}
\label{subsec:convergence}

\begin{figure}[!t]\centering
	\includegraphics[width=0.98\columnwidth]{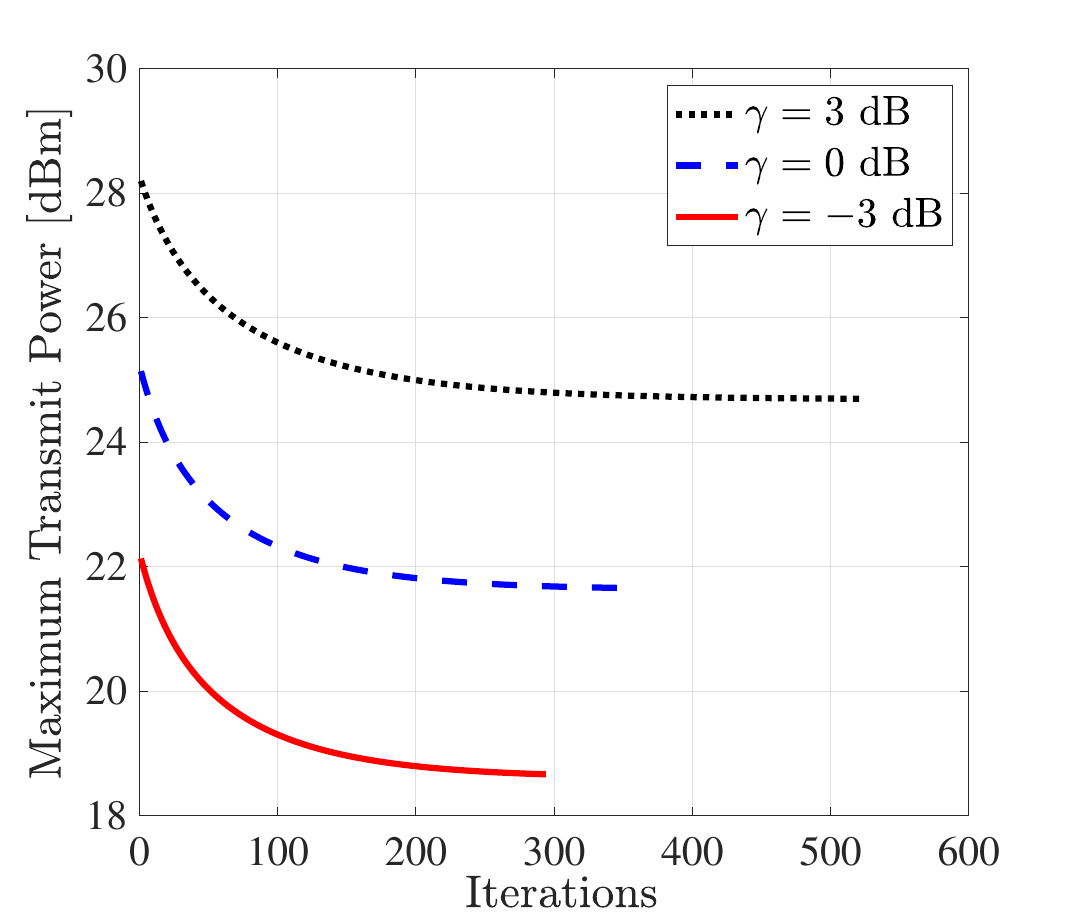}
	\caption{
		Convergence behavior for $N_b = 32$ antennas, $N_c = 2$ cells, $N_u = 3$ users per cell,  $b =3$ bits, $K=64$ subcarriers, and $\gamma \in\{-3, 0, 3\}\ \rm dB$.
		} 
	\label{fig:convergence}
\end{figure}

Fig.~\ref{fig:convergence} shows progress in maximum transmit power with respect to the number of iterations of Q-CoMP-PA until the stopping condition is met considering $\gamma \in\{-3, 0, 3\}\ \rm dB$, $b= 3$, $N_b = 32$, $N_c = 2$, $N_u = 3$, and $K=64$. 
Note that we plot $T_0T_3$ defined in the complexity analysis in Section~\ref{sec:complexity}.
It can be seen that more iterations are needed and the algorithm converges to a higher peak power as the target SQINR increases.

\begin{figure}[!t]\centering
	\includegraphics[width=1\columnwidth]{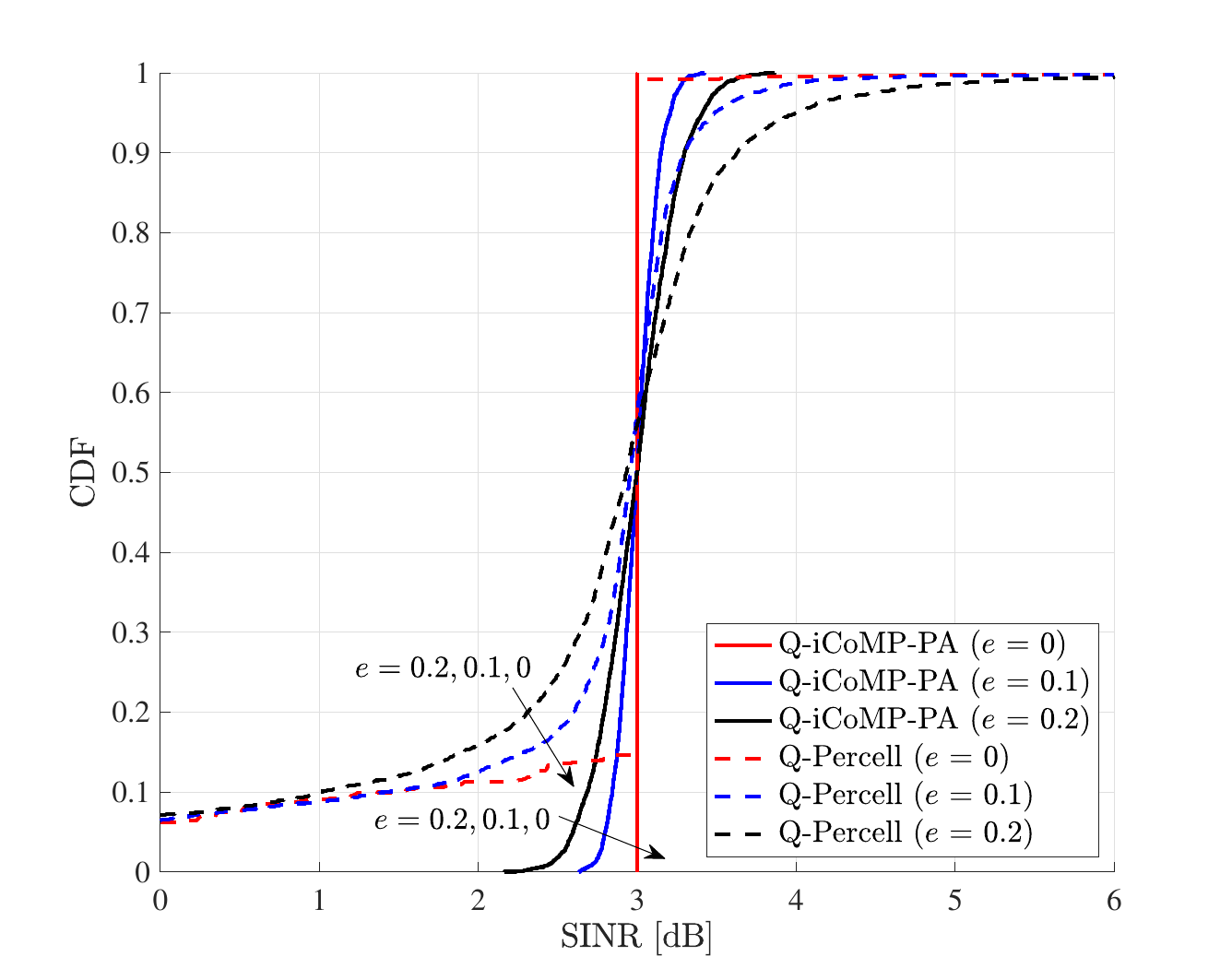}
	\caption{CDFs of the achieved SQINRs for $\gamma = 3\  \rm dB$ target SQINR, $b= 3$ bits, $N_b = 32$ antennas, $N_c = 4$ cells, and $N_u = 3$ users per cell with channel error factors $e \in \{0,0.1,0.2\}$.
	} 
	\label{fig:herror}
\end{figure}

\subsection{Channel Estimation Error}
To assess the effects of channel estimation error, we employ a Gaussian channel estimation error model \cite{wang2006adaptive} that defines the estimated UL channel between BS $i$ and user $u$ in cell $j$ as 
\begin{equation}
    \hat{\bh}_{i,j,u} = \sqrt{(1-e^2)}\;\bh_{i,j,u} + e\;{\be}_{i,j,u},
\end{equation}
where $e$ dictates the channel estimation error factor, and $\be_{i,j,u}$ is an error vector whose elements follow $\mathcal{CN}(0, \rho_{i,j,u})$ in which $\rho_{i,j,u}$ denotes large scale fading between BS$_i$ and user $u$ in cell $j$.

Fig.~\ref{fig:herror} shows CDFs of the achieved SQINR for channel estimation error factor $e \in \{0, 0.1, 0.2\}$, $\gamma = 3\  \rm dB$, $b= 3$, $N_b = 32$, $N_c = 4$, and $N_u = 3$.
As $e$ increases, the deviation from the target SQINR increases, resulting in less accurate performance in  achieved SQINR.
The proposed Q-CoMP-PA algorithm is more robust to the estimation error than the Q-Percell method, showing a smaller deviation from the target SQINR $\gamma = 3\  \rm dB$.

\begin{table}[!t]
\caption{Comparison of peak-to-average power ratio (PAPR) for selected target SQINR values.}
\begin{center}
\begin{tabular}{|c|c|c|c|}
	\hlineB{2}
	$\gamma$ [dB] &  Q-Percell & Q-CoMP & {\bf Q-CoMP-PA} \\
	\hline
	1 & 8.21 dB & 4.98 dB & {\bf 2.96 dB} \\
	\hline
	-5 & 7.99 dB & 4.93 dB & {\bf 2.88 dB} \\
	\hlineB{2}
\end{tabular}
\\
\vspace{0.5cm}
(a) PAPR of the wideband system with $N_b = 32$ antennas, $N_c = 3$ cells, $N_u = 2$ users per cell,  $b =3$ bits, and $K=64$.\\

\vspace{0.3cm}

\begin{tabular}{|c|c|c|c|}
\hlineB{2}
$\gamma$ & \multicolumn{3}{c|}{{$N_c=2$ cells}}\\
\cline{2-4} 
[dB] & Q-Percell&  Q-CoMP & {\bf Q-CoMP-PA} \\
\hline
2 & 4.02 dB & 2.96 dB & {\bf 2.18 dB}\\
\hline
-3 & 3.97 dB & 2.94 dB & {\bf 2.07 dB}\\
\hlineB{2}
\end{tabular}
\vspace{0.3cm}

\begin{tabular}{|c|c|c|c|}
\hlineB{2}
$\gamma$ & \multicolumn{3}{c|}{{$N_c=4$ cells}}\\
\cline{2-4}
[dB] & Q-Percell&  Q-CoMP & {\bf Q-CoMP-PA} \\
\hline
2 & 5.35 dB & 4.21 dB & {\bf2.25 dB}\\
\hline
-3 & 5.58 dB & 4.42 dB & {\bf2.55 dB}\\
\hlineB{2}
\end{tabular}
\vspace{0.3cm}

\begin{tabular}{|c|c|c|c|}
\hlineB{2}
$\gamma$ & \multicolumn{3}{c|}{{$N_c=6$ cells}}\\
\cline{2-4}
[dB] & Q-Percell&  Q-CoMP & {\bf Q-CoMP-PA} \\
\hline
2 & 5.49 dB & 4.57 dB & {\bf2.51 dB}\\
\hline
-3 & 5.96 dB & 4.81 dB & {\bf3.02 dB}\\
\hlineB{2}
\end{tabular}
\\
\vspace{0.5cm}
(b) PAPR of the narrowband system with with $N_b = 32$ antennas, $N_c \in \{2,4,6\}$ cells, $N_u = 2$ users per cell, and $b =3$ bits.

\label{table:papr}
\end{center}
\end{table}

\subsection{Peak-to-Average-Power Ratio}
In this subsection, we simulate the PAPR of the transmitted signals.
Since a large peak amplitude causes severe performance loss, reducing the PAPR is considered as one of the most crucial topics in OFDM systems.
In Table~\ref{table:papr}(a), we analyze a wideband communication network  with $N_b = 16$, $N_c=3$,  $N_u = 2$, $K=64$, $b=3$, and $\gamma\in\{-5,1\}\ \rm dB$.
It is observed that  Q-CoMP-PA achieves the lowest PAPR for the considered target SQINRs, reducing it by more than $2 \ \rm dB$ and $5 \ \rm dB$ from Q-CoMP and Q-Percell, respectively.
Table~\ref{table:papr}(b) considers a narrowband system with $N_b = 32$, $N_c \in \{2,4,6\}$,  $N_u = 2$, $b=3$, and $\gamma\in\{-3,2\} \ \rm dB$.
Q-CoMP-PA still achieves significant reduction in PAPR, showing more than $1.8 \ \rm dB$ gain over Q-CoMP. 
It is also shown that Q-CoMP-PA maintains similar PAPR over different SQINR targets or the different number of cells.
 The largest PAPR is required by Q-Percell with limited multicell coordination, highlighting the necessity of proper multicell coordination.
Thus, Q-CoMP-PA is more favorable for communication systems by limiting the peak power.

\section{Conclusion}
\label{sec:conclusion}

We investigated the DL OFDM CoMP beamforming problem with low-resolution quantizers and employed per-antenna power constraints to give a more practical power-efficient solution.
The DL antenna power minimization problem with SQINR constraints was formulated as a primal problem, and we derived the associated dual problem.
We interpreted the dual problem as a virtual UL power minimization problem with respect to the transmit power and noise covariance matrix.
Subsequently, we showed that strong duality holds between the primal DL and dual UL problems.
Inspired by the strong duality, we proposed an iterative DL beamforming algorithm.
To this end, we addressed the inner and outer problems of the dual in an alternating manner by solving the inner power control problem via fixed-point iteration and the outer noise covariance design problem via projected subgradient ascent.
Then, the DL beamforming solutions were directly obtained from the UL solutions through linear transformation.
In simulation, we demonstrated the proposed Q-CoMP-PA algorithm is effective by reducing per-antenna transmit power while achieving target requirements over a wide variety of system settings.
The derived duality can provide insight for designing  power-efficient communication networks and the proposed method can contribute to realizing highly efficient and reliable future communication networks.

\begin{appendices}
\renewcommand{\thesectiondis}[2]{\Alph{section}:}

\section{Proof of Theorem~\ref{thm:duality_ofdm}}
\label{appx:duality_ofdm}

    We show that the Lagrangian dual problem of the primal DL problem in \eqref{eq:dl_problem}-\eqref{eq:per_antenna_const} is equivalent to the virtual UL power minimization problem. 
    We first acquire the Lagrangian dual problem of \eqref{eq:dl_problem}-\eqref{eq:per_antenna_const}.
    To this end, we rewrite the per-antenna constraint for  antenna $m$ in \eqref{eq:per_antenna_const}  as
    \begin{align}
    	 \Big[\bbE[\bx_{{\rm q},i}(k)\bx_{{\rm q},i}^H(k)]\Big]_{m,m} 
    	  \!\!=\! \left[\frac{\alpha}{K}\!\sum\nolimits_{\ell=0}^{K-1}\!\bW_i(\ell)\bW_i^H(\ell)\right]_{m,m}\!.
    	  \label{eq:per_antenna_const2}
    \end{align}
    We multiply the objective function in \eqref{eq:dl_problem} by $KN_cN_b$ without affecting the solution.
    With non-negative Lagrangian multipliers $\mu_{i,u}(k)$ and $\nu_{i,m}(k)$, the Lagrangian of the DL problem becomes 
    \begin{align}
         &{\cL\left(\bw_{i,u}(k),\mu_{i,u}(k),\nu_{i,m}(\ell)\right)}
         =  \, KN_cN_b p_{\rm 0} \label{eq:lagrangian}\\
         &- \sum\nolimits_{i,u,k}\mu_{i,u}(k)\Big(\alpha^2|\bw_{i,u}^H(k)\bg_{i,i,u}(k)|^2/\gamma_{i,u}(k)  \nonumber \\ 
         &-\alpha^2 \!\!\!\!\! \sum_{(j,v) \neq (i,u)}^{N_c, N_u}\!\!|\bg_{j,i,u}^H(k)\bw_{j,v}(k)|^2  - {\rm Q}_{i,u}(k) - \sigma^2\Big) 
         \nonumber\\
         &+  \sum_{i,m,\ell}\nu_{i,m}(\ell) \!\!
         \left(\left[\frac{\alpha}{K}\sum\nolimits_{\ell=0}^{K-1}\bW_i(\ell)\bW_i^H(\ell)\right]_{m,m}
        \!\!\!\!\!\!- p_{\rm 0}   \right). \nonumber
    \end{align}
    
    To make the Lagrangian tractable, we start by rewriting $\sum_{i,u,k}\mu_{i,u}(k){\rm Q}_{i,u}(k)$  to manipulate $\bW_i$ which is embedded in $\bC_{{\bf q}_i}$ of ${\rm Q}_{i,u}(k)$. 
    By mapping the indices from $(i,u,k,j)$ to $(j,v,\ell,i)$, we have
    \begin{align}
        &\sum_{i,u,k}\mu_{i,u}(k){\rm Q}_{i,u}(k)\\
        &=\sum_{i,u,k}\mu_{i,u}(k)\sum_{j=1}^{N_c}\ubg_{j,i,u}^H\!(k){\pmb \Psi}_{N_b}\bC_{\bq_j,\bq_j}{\pmb \Psi}_{N_b}^H\ubg_{j,i,u}(k)
        \\ \label{eq:quant_noise_dl_ofdm1}
        &=\sum_{j,v,\ell,i}\mu_{j,v}(\ell)\ubg_{i,j,v}^H\!(\ell){\pmb \Psi}_{N_b}   \times \\
        &\quad{\rm diag}\bigg(\sum\nolimits_{u,k}|{\pmb \psi}_{N_b,m}^H(n)\ubw_{i,u}(k)|^2, \forall m,n\!\bigg){\pmb \Psi}_{N_b}^H\ubg_{i,j,v}(\ell) \nonumber
    \end{align}
    where ${\pmb \psi}_{N_b,m}(n)$ is the $\left(m+(n-1)N_b\right)$th column of ${\pmb \Psi}_{N_b}$, i.e., ${\pmb \psi}_{N_b,m}(n) = [\bw_{{\rm DFT},n}\otimes\bI_{N_b}]_{:,m}$ for $m= 1,\dots,N_b$, $n=1,\dots,K$,   and $\ubw_{i,u}(k)$ is the $(kN_u+u)$th column of $\bW_i$.
    Let  $\bM_i(k) = {\rm diag}(\mu_{i,1}(k),\dots,\mu_{i,N_u}(k))$, $\bM_i = {\rm blkdiag}\big(\bM_i(0),\dots,\bM_i(K-1)\big)$, and $\bM = [\bM_{1},\dots, \bM_{N_c}]$.
    Recalling that ${\pmb \Psi}_{N_b}(k) = \big([\bW_{{\rm DFT}}]_{k+1,:}\otimes \bI_{N_b}\big)$ and $\bG_i = [\bG_{i,1},\dots, \bG_{i,N_c}]$,  \eqref{eq:quant_noise_dl_ofdm1} is rewritten as 
\begin{align}
        \nonumber
        &\sum_{j,v,\ell,i}\mu_{j,v}(\ell) \sum_{m,n}\Bigg(\sum_{u,k}|{\pmb \psi}_{N_b,m}^H(n)\ubw_{i,u}(k)|^2\times\\ &\left(\sum_{r}{\underline{g}}^*_{i,j,v,r}(\ell)\psi_{N_b,m,r}(n)\right)\!\!\left(\sum_{r'}{\underline{g}}_{i,j,v,r'}(\ell)\psi^*_{N_b,m,r'}(n)\right)\Bigg)\nonumber\\
        \nonumber
        & = \sum_{i,u,k}\ubw_{i,u}^H(k)\Bigg(\sum_{m,n}{\pmb \psi}_{N_b,m}(n)\Bigg(\sum_{j,v,\ell}\mu_{j,v}(\ell){\pmb \psi}_{N_b,m}^H(n)\times\\
        &\qquad\ubg_{i,j,v}(\ell){\ubg}^H_{i,j,v}(\ell){\pmb \psi}_{N_b,m}(n)\Bigg)\!{\pmb \psi}_{N_b,m}^H(n)\!\Bigg)\ubw_{i,u}(k)\\
        \nonumber
        & = \sum\nolimits_{i,u,k}\ubw_{i,u}^H(k){\pmb \Psi}_{N_b}{\rm diag}\left({\pmb \Psi}_{N_b}^H\bG_i\bM\bG_i^H{\pmb \Psi}_{N_b}\right){\pmb \Psi}_{N_b}^H\ubw_{i,u}(k)
        \nonumber\\
        \label{eq:thm_proof1}
        & \stackrel{(a)}= \sum\nolimits_{i,u,k}\bw_{i,u}^H(k){\pmb \Psi}_{N_b}(k){\rm diag}\Big({\pmb \Psi}_{N_b}^H\bG_i\bM\bG_i^H{\pmb \Psi}_{N_b}\Big)\times\nonumber \\
        &\qquad\qquad{\pmb \Psi}_{N_b}^H(k)\bw_{i,u}(k).
\end{align}
   Here, $(a)$ comes from $\ubw_{i,u}^H(k){\pmb \Psi}_{N_b} = \bw_{i,u}^H(k){\pmb \Psi}_{N_b}(k)$ since $\ubw_{i,u}(k)$ has nonzero elements $\bw_{i,u}(k)$ only in the position that corresponds to the precoder for subcarrier $k$, and ${\underline{g}}_{i,j,v,r}(\ell)$ and ${\psi}_{N_b,m,r}(n)$ are the $r$th elements of $\ubg_{i,j,v}(\ell)$ and ${\pmb \psi}_{N_b,m}(n)$, respectively.

    Next, we define $\tilde\bD_i(k) = {\rm diag}(\nu_{i,1}(k), \dots, \nu_{i,N_b}(k))$. 
    By switching the indices between $k$ and $\ell$, we can cast $\sum_{i,m,k}\nu_{i,m}(k)\left[\sum_{\ell=0}^{K-1}\bW_i(\ell)\bW_i^H(\ell)\right]_{m,m}$  in \eqref{eq:lagrangian} to
    \begin{align}
    	\nonumber
    	&\sum\nolimits_{i,m,\ell}\nu_{i,m}(\ell)\bigg[\sum\nolimits_{u,k}\bw_{i,u}(k)\bw_{i,u}^H(k)\bigg]_{m,m} \\
     &=  \sum\nolimits_{i,m,\ell}\sum\nolimits_{u,k} w^*_{i,u,m}(k)\nu_{i,m}(\ell)w_{i,u,m}(k)\\
    	\label{eq:thm_proof2}
    	& = \sum\nolimits_{i,u,k}\bw^H_{i,u}(k)\tilde\bD_i\bw_{i,u}(k),
   	\end{align}
    where $w_{i,u,m}(k)$ is the $m$th element of $\bw_{i,u}(k)$ and $\tilde\bD_i = \sum_{\ell=0}^{K-1}\tilde \bD_i(\ell)$.
    Further, $\sum_{i,m,\ell}\nu_{i,m}(\ell)p_{\rm 0}$ in \eqref{eq:lagrangian}  can be redrafted as
    \begin{align}
    	\label{eq:thm_proof3}
    	p_{\rm 0}\sum\nolimits_{i,m,\ell}\nu_{i,m}(\ell) = p_{\rm 0}\sum\nolimits_{i=1}^{N_c}{\rm tr}(\tilde\bD_i).
    \end{align}

    By applying \eqref{eq:thm_proof1}, \eqref{eq:thm_proof2}, and \eqref{eq:thm_proof3} to the Lagrangian in \eqref{eq:lagrangian}, the Lagrangian is reformulated as
    \begin{align}	
     \label{eq:lagrangian_reform}
        	&{\cL\left(\bw_{i,u}(k),\mu_{i,u}(k),\nu_{i,m}(\ell)\right)}
        	=  \\
         &\sum\nolimits_{i,u,k}\mu_{i,u}(k)\sigma^2- p_{\rm 0}\sum\nolimits_{i}\big[{\rm tr}(\tilde\bD_i)-KN_b\big]
        	\nonumber \\
         \nonumber
&+\sum\nolimits_{i,u,k}\bw_{i,u}^H(k)\Bigg(\frac{\alpha}{K}\tilde\bD_{i} - \alpha^2\Big(1+1/\gamma_{i,u}(k)\Big)\times 
\\
&\mu_{i,u}(k)\bg_{i,i,u}(k)\bg_{i,i,u}^H(k)\!+\!\alpha^2\sum\nolimits_{j,v}\mu_{j,v}(k)\bg_{i,j,v}(k)\bg_{i,j,v}^H(k)\nonumber\\
&\!+\!\alpha(1-\alpha){\pmb \Psi}_{N_b}(k){\rm diag}\!\left({\pmb \Psi}_{N_b}^H\bG_i\,\bM\,\bG_i^H{\pmb \Psi}_{N_b}\right){\pmb \Psi}_{N_b}^H(k)\Bigg)\bw_{i,u}(k). \nonumber
   	 	\end{align}
    Let us define the dual objective function as 
    \begin{align}
        \label{eq:dual_objective}
        g(\mu_{i,u}(k)),\nu_{i,m}(\ell)) = \min_{{\bf w}_{i,u}(k),p_{\rm 0}} {\cL(\bw_{i,u}(k),\mu_{i,u}(k),\nu_{i,m}(\ell))}
    \end{align}
    and let $\bD_i = \frac{1}{K}\tilde\bD_i$. 
   	To avoid an unbounded objective function, we need ${\rm tr}(\bD_i)\leq N_b$ and
   	\begin{align}
       	\label{eq:K_matrix2}
        {\bf K}_{i,k}(\bM) &\!=\! {\bf D}_{i} \! + \!\alpha \!\sum_{j,v} \!\mu_{j,v}(k){\bf g}_{i,j,v}(k){\bf g}_{i,j,v}^H(k) \\
        &\!+\! (1\!-\!\alpha){\pmb \Psi}_{\!N_b}\!(k){\rm diag}\!\Big(\!{\pmb \Psi}_{N_b}^H\bG_i \bM \bG_i^H{\pmb \Psi}_{\!N_b}\!\Big){\pmb \Psi}_{\!N_b}^H\!(k) \nonumber
        \\
         &\succeq \alpha \big(1 + 1/{\gamma_{i,u}(k)}\big)\mu_{i,u}(k)  {\bf g}_{i,i,u}(k){\bf g}_{i,i,u}^H(k).
   	\end{align}
    Consequently, the Lagrangian dual problem of \eqref{eq:dl_problem} can be eventually formulated as
    \begin{align}
        \label{eq:dual_problem_reduced_inv}
        \max_{{\bD}_i }\max_{\mu_{i,u}(k)}& \;\;  \sum\nolimits_{i,u,k}\mu_{i,u}(k)\sigma^2\\
        \label{eq:dual_SQINR_const_reduced_inv}
        {\text{ \rm subject to}} & \;\; {\bf K}_{i,k}(\bM) \succeq  \\
        &\;\; \alpha \big(1+1/\gamma_{i,u}(k)\big)\mu_{i,u}(k)  {\bf g}_{i,i,u}(k){\bf g}_{i,i,u}^H(k),\nonumber \\
        &\;\;  \bD_i \succeq 0, \ \bD_i\in \bbR^{N_b\times N_b}: \text{\rm diagonal}, \\
        &\;\;  {\rm tr}(\bD_i) \leq N_b \quad \forall i,u,k.
   	\end{align}

Now, let us show that the problem in \eqref{eq:dual_problem_reduced_inv} is equivalent to the problem in Theorem~\ref{thm:duality_ofdm} which is the virtual UL transmit power minimization problem.
To this end, we use $\mu_{i,u}(k)$ and $\lambda_{i,u}(k)$  interchangeably. 
Then ${\bf K}_{i,k}({\bf M}) = {\bf K}_{i,k}({\pmb \Lambda})$ in \eqref{eq:K_matrix2} can be considered as the covariance matrix of the UL received signal scaled by $1/\alpha$ in compact form.
Accordingly, we can rewrite the covariance matrix of the interference-plus-quantization-plus-nose term as $\bZ_{i,u}(k) = \alpha{\bf K}_{i,k}({\pmb \Lambda})-\alpha^2\lambda_{i,u}(k) {\bf g}_{i,i,u}(k){\bf g}_{i,i,u}^H$. 
From this, we find that the constraint in \eqref{eq:dual_SQINR_const_reduced_inv} is satisfied when $\gamma_{i,u}(k)\bZ_{i,u}(k) \succeq   \alpha^2 \lambda_{i,u}(k) {\bf g}_{i,i,u}(k){\bf g}_{i,i,u}^H$. 
This condition  can be further led to 
\begin{align}
        0\leq&{\bf g}_{i,i,u}^H(k)\bigg(\gamma_{i,u}(k) \bI_{N_b} \nonumber\\
        &\qquad- \alpha^2 \lambda_{i,u}(k) {\bf g}_{i,i,u}(k){\bf g}_{i,i,u}^H \bZ_{i,u}^{-1}(k)\bigg){\bf g}_{i,i,u}(k) \\
        =&{\bf g}_{i,i,u}^H(k){\bf g}_{i,i,u}(k)\bigg(\gamma_{i,u}(k) \nonumber\\
        &\qquad- \alpha^2\lambda_{i,u}(k){\bf g}_{i,i,u}^H(k)\bZ_{i,u}^{-1}(k){\bf g}_{i,i,u}(k)\bigg), 
    \end{align}
    which is equivalent to $\alpha^2 \lambda_{i,u}(k) {\bf g}_{i,i,u}^H \bZ_{i,u}^{-1}(k){\bf g}_{i,i,u}(k) \leq \gamma_{i,u}(k)$.

    If we consider the uplink SQINR constraint in \eqref{eq:ul_const} in a reversed way as 
\begin{equation}
    \label{eq:ul_const_reverse}
    \max_{\bff_{i,u}(k)}{\Gamma}^{\rm ul}_{i,u}(k) \leq \gamma_{i,u}(k) \;\;\forall i,u,k,
\end{equation}
    applying the MMSE combiner in \eqref{eq:mmse} and ${\pmb \Psi}_{N_b}(k){\pmb \Psi}_{N_b}^H(k) = \bI_{N_b}$ to \eqref{eq:ul_const_reverse} can  
    reduce \eqref{eq:ul_const_reverse} to $\alpha^2 \lambda_{i,u}(k) {\bf g}_{i,i,u}^H \bZ_{i,u}^{-1}(k){\bf g}_{i,i,u}(k) \leq \gamma_{i,u}(k)$.
    Therefore, we can claim that the condition \eqref{eq:dual_SQINR_const_reduced_inv} is identical to ${\Gamma}^{\rm ul}_{i,u}(k) \leq \gamma_{i,u}(k)$ when $\bff_{i,u}(k) = \bff_{i,u}^{\sf MMSE}(k)$.
    In other words, we have shown that the problem in \eqref{eq:dual_problem_reduced_inv} is the same as
    \begin{align}
        \label{eq:dual_problem_reduced_inv2}
        \max_{{\bD}_i }\max_{\lambda_{i,u}(k)}& \;\;  \sum\nolimits_{i,u,k}\lambda_{i,u}(k)\sigma^2\\
        \label{eq:dual_SQINR_const_reduced_inv2}
         {\text{ \rm subject to}} & \;\;{\Gamma}^{\rm ul}_{i,u}(k) \leq \gamma_{i,u}(k),\\
        &\;\; \bff_{i,u}(k) = \bff_{i,u}^{\sf MMSE}(k),\\
        \nonumber
        &\;\;  \bD_i \succeq 0, \ \bD_i\in \bbR^{N_b\times N_b}: \text{\rm diagonal}, \\
        &\;\;  {\rm tr}(\bD_i) \leq N_b \quad \forall i,u,k.
   	\end{align}
    
    The rest of the proof aims to show that the reformulated Lagrangian dual in \eqref{eq:dual_problem_reduced_inv2} 
    is equivalent to the following:
    \begin{align}
        \label{eq:dual_problem_reduced}
        \max_{{\bD}_i }\min_{\lambda_{i,u}(k)}& \;\;  \sum\nolimits_{i,u,k}\lambda_{i,u}(k)\sigma^2\\
        \label{eq:dual_SQINR_const_reduced}
        {\text{ \rm subject to}} & \;\;{\Gamma}^{\rm ul}_{i,u}(k) \geq \gamma_{i,u}(k),\\
        &\;\; \bff_{i,u}(k) = \bff_{i,u}^{\sf MMSE}(k),\\
        \nonumber
        &\;\;  \bD_i \succeq 0, \ \bD_i\in \bbR^{N_b\times N_b}: \text{\rm diagonal}, \\
        \nonumber
        &\;\;  {\rm tr}(\bD_i) \leq N_b \quad \forall i,u,k.
   	\end{align}
    The key differences between \eqref{eq:dual_problem_reduced_inv2} and \eqref{eq:dual_problem_reduced} are the opposite objective problems with respect to $\lambda_{i,u}(k)$, i.e., $\max$ vs. $\min$, and the reversed SQINR conditions in \eqref{eq:dual_SQINR_const_reduced_inv2} and \eqref{eq:dual_SQINR_const_reduced}. 
    Here, \eqref{eq:dual_SQINR_const_reduced} corresponds to the minimum virtual UL SQINR constraints for the minimizing problem with respect to $\lambda_{i,u}(k)$
    whereas \eqref{eq:dual_SQINR_const_reduced_inv2} denotes the maximum SQINR constraints for the maximization problem with respect to $\lambda_{i,u}(k)$.
    To relate the two problems, the optimality conditions are summarized in the following proposition, 
    \begin{proposition}
        \label{prop:1}
    The minimum objective with minimum SQINR constraints, i.e., \eqref{eq:dual_problem_reduced}-\eqref{eq:dual_SQINR_const_reduced} and the maximum objective with maximum SQINR constraints, i.e., \eqref{eq:dual_problem_reduced_inv2}-\eqref{eq:dual_SQINR_const_reduced_inv2} have their optimal solutions when the SQINR constraints are satisfied with equality.
    \begin{proof}
        Without loss of generality, we focus on the $u$th user in $i$th cell at $k$th subcarrier.
        Suppose that \eqref{eq:dual_problem_reduced} achieves the optimal solution while \eqref{eq:dual_SQINR_const_reduced} does not claim equality, i.e., ${\Gamma}^{\rm ul}_{i,u}(k) > \gamma_{i,u}(k)$.
        In this case, we can further decrease the UL SQINR by making $\lambda_{i,u}(k)$ smaller until equality condition holds, and hence the objective function can diminish, which leads to a contradiction.
        The case of \eqref{eq:dual_problem_reduced_inv2}-\eqref{eq:dual_SQINR_const_reduced_inv2} can be similarly proved.
    \end{proof}
    \end{proposition}
    From Proposition~\ref{prop:1}, we know that the optimal solutions for both problems are achieved with the same active SQINR constraints, and hence, the two problems are indeed equivalent with the identical objective function.
    Therefore, \eqref{eq:dual_problem_reduced_inv2} and \eqref{eq:dual_problem_reduced} are equivalent since  
    $\max$ of the inner problem can be replaced with $\min$, and both the SQINR inequality constraints can be replaced with equality constraints.    
    This implies that the Lagrangian dual of the DL problem can be considered as the virtual UL power minimization problem with outer maximization over $\bD_i$, which leads to an efficient solution in the sequel.
    This completes the proof of Theorem~1.
    \qed

\section{Proof of Corollary~\ref{cor:strongduality}}
\label{appx:strongduality}
First of all, the identity in \eqref{eq:per_antenna_const2} allows the primal DL problem in \eqref{eq:dl_problem} to be rewritten as
    \begin{gather}
        \label{eq:strong_pf}
        \min_{{\bw}_{i,u}(k),p_o} p_o \\ 
        \label{eq:strong_pf1}
        {\rm s.t.}\ \Gamma_{i,u}(k) \geq \gamma_{i,u}(k), \quad \forall i,u,k\\ 
        \label{eq:strong_pf2}
        \left[\frac{\alpha}{K}\sum\nolimits_{\ell=0}^{K-1}\bW_i(\ell)\bW_i^H(\ell)\right]_{m,m} \leq p_o.
     \end{gather}
    Let $\bW_{\rm BD}(k) = {\rm blkdiag}(\bW_1(k),\dots,\bW_{N_c}(k))$, $\tilde{\bW}_{\rm BD}(k) = {\rm blkdiag}((\bI_{KN_b}\otimes\bW_1(k)),\dots,(\bI_{KN_b}\otimes\bW_{N_c}(k)))$, $\tilde{\bW}_{\rm BD} = {\rm blkdiag}(\tilde{\bW}_{\rm BD}(0),\ldots, \tilde{\bW}_{\rm BD}(K-1))$, $\tilde{{\pmb \Psi}}_{N_b}(k)=\bI_{KN_b}\otimes{\pmb \Psi}_{N_b}(k)$, $\tilde{{\pmb \Psi}}_{N_b} = {\rm blkdiag}(\tilde{{\pmb \Psi}}_{N_b}(0),\ldots, \tilde{{\pmb \Psi}}_{N_b}(K-1))$, $\bE_{j,i,u}(k) = {\rm diag}({\pmb \Psi}_{N_b} \bg_{j,i,u}\bg_{j,i,u}^H{\pmb \Psi}_{N_b}^H)$, and $\bE_{i,u}(k) = {\bf 1}_K\otimes {\rm vec}(\bE^{1/2}_{1,i,u}(k),\ldots,\bE^{1/2}_{N_c,i,u}(k))$. 
    The SQINR constraints in \eqref{eq:strong_pf1} can be re-interpreted as
    \begin{align}
        \label{eq:strong_pf4}
        &\alpha^2\left(1+1/\gamma_{i,u}(k)\right) |{\bf w}_{i,u}^H(k) {\bf g}_{i,i,u}|^2 \\
        & \geq
        \left\| 
        \begin{matrix}
            \alpha\bW_{\rm BD}^H(k){\rm vec}(\bg_{1,i,u},\dots,\bg_{N_c,i,u}) \\
           \sqrt{\alpha(1-\alpha)}\tilde{\bW}_{\rm BD}(0)\tilde{{\pmb \Psi}}_{N_b}(0) \times \\
           {\rm vec}(\bE^{1/2}_{1,i,u}(k),\!..,\bE^{1/2}_{N_c,i,u}(k))\\
            \vdots \\
            \sqrt{\alpha(1-\alpha)}\tilde{\bW}_{\rm BD}(K-1)\tilde{{\pmb \Psi}}_{N_b}(K-1) \times\\
            {\rm vec}(\bE^{1/2}_{1,i,u}(k),\!..,\bE^{1/2}_{N_c,i,u}(k)) 
        \end{matrix}
        \right\|^2 + \sigma^2 \\
        &=  \left\|
        \begin{matrix}
            \alpha\bW_{\rm BD}^H(k){\rm vec}(\bg_{1,i,u},\dots,\bg_{N_c,i,u}) \\
            \sqrt{\alpha(1-\alpha)} \tilde{\bW}_{\rm BD} \tilde{{\pmb \Psi}}_{N_b} \bE_{i,u}(k)
        \end{matrix}
        \right\|^2 + \sigma^2,
    \end{align}
    for all $i$, $u$, and $k$. In addition, the per-antenna constraint in \eqref{eq:strong_pf2} is rewritten as
    \begin{align}
        &\left[\frac{\alpha}{K}\sum\nolimits_{\ell=0}^{K-1}\bW_i(\ell)\bW_i^H(\ell)\right]_{m,m}
        \nonumber \\
        & = \frac{\alpha}{K} \left\|{\rm vec}(\be_{m}^H\bW_i(0),\ldots, \be_{m}^H\bW_i(K-1)) \right\|^2.
    \end{align}
	for all $m$. 
    Accordingly,  the primal DL problem in \eqref{eq:dl_problem} can be cast to the SOCP. 
	Since \eqref{eq:dl_problem} is strictly feasible and convex, 
	strong duality holds between \eqref{eq:dl_problem} and \eqref{eq:dual_problem}. 
\qed

\section{Proof of Corollary~\ref{cor:dl_precoder_ofdm}}
\label{appx:dl_precoder_ofdm}
We first find the derivative of the Lagrangian in  \eqref{eq:lagrangian_reform} with respect to $\bw_{i,u}(k)$ as
\begin{align}
    \label{eq:derivative}
    &\frac{\partial {\cL(\bw_{i,u}(k),\lambda_{i,u}(k),\nu_{i,m}(\ell))}}{\partial \bw_{i,u}(k)}
    \\\nonumber
    &=2\Big(\frac{\alpha}{K}\tilde\bD_{i}\! -\! \alpha^2\Big(1+1/\gamma_{i,u}(k)\Big)\lambda_{i,u}(k)\bg_{i,i,u}(k)\bg_{i,i,u}^H(k)  
    \\\nonumber
    &\quad +\alpha^2\!\sum\nolimits_{j,v}\lambda_{j,v}(k)\bg_{i,j,v}(k)\bg_{i,j,v}^H(k)
    \\\nonumber
    &\quad +\!\alpha(1-\alpha){\pmb \Psi}_{N_b}(k){\rm diag}\!\left({\pmb \Psi}_{N_b}^H\bG_i\,{\pmb\Lambda}\,\bG_i^H{\pmb \Psi}_{N_b}\right)\!{\pmb \Psi}_{N_b}^H(k)\!\Big)\bw_{i,u}(k).
\end{align}
We then set the derivative of the Lagrangian to zero, and solve it for $\bw_{i,u}(k)$ as
\begin{align}
    \nonumber
    \bw_{i,u}(k) &=  \bigg( \!\alpha^2\!\!\!\!\!\sum_{(j,v)\neq (i,u)} \!\!\!\lambda_{j,v}(k) {\bf g}_{i,j,v}(k){\bf g}_{i,j,v}^H(k)\! 
    \\\nonumber
    &\quad +\! \alpha(1\!-\!\alpha){\pmb \Psi}_{N_b}(k){\rm diag}\!\left({\pmb \Psi}_{N_b}^H\bG_i\,{\pmb\Lambda}\,\bG_i^H{\pmb \Psi}_{N_b}\right)\!{\pmb \Psi}_{N_b}^H(k)\! 
    \\\nonumber
    & \quad +\! \alpha \bD_i\! \bigg)^{-1}\!\!\frac{\alpha^2}{\gamma_{i,u}(k)}\!\lambda_{i,u}(k){\bf g}_{i,i,u}(k)\bg_{i,i,u}^H(k)\bw_{i,u}(k)   
    \\\nonumber
    & \stackrel{(a)}= \left(\alpha^2/\gamma_{i,u}(k)\right)\lambda_{i,u}(k)\bg_{i,i,u}^H(k)\bw_{i,u}(k) \bff_{i,u}^{\sf MMSE}(k)
\end{align}
where $(a)$ uses the MMSE equalizer $\bff_{i,u}^{\sf MMSE}(k)$ defined in \eqref{eq:mmse}. 
As a result, we argue $\bw_{i,u}(k) = \sqrt{\tau_{i,u}(k)}\bff_{i,u}^{\sf MMSE}(k)$ with properly computed $\tau_{i,u}(k)$.

To satisfy the Karush–Kuhn–Tucker stationarity condition with the DL constraint in \eqref{eq:dl_SQINR}, $\Gamma_{i,u}(k)$ has to meet the target SQINR constraint with strict equality.
For given $i,u,k$, let us define $\xi_{i',u'}(n)=1$ if $(i',u',n)=(i,u,k)$, and $\xi_{i',u'}(n)=0$ otherwise. 
We collect $\xi_{i,u}$'s as $\Xi_i(k) = {\rm diag}(\xi_{i,1}(k),\dots,\xi_{i,N_u}(k))$, ${\Xi(k)} = {\rm blkdiag}({\Xi}_1,\dots, {\Xi}_{N_c})$, and ${{\Xi}} = {\rm blkdiag}({\Xi(0)}\dots, {\Xi}(K-1))$.
We then rewrite the quantization error term in \eqref{eq:SQINR_dl_ofdm} in a tractable form as
\begin{align}
    \nonumber
    &{\rm Q}_{i,u}(k)
    \\\nonumber
    &\stackrel{(a)}= \alpha(1-\alpha)\sum_{j=1}^{N_c}\bg_{j,i,u}^H\!(k){\pmb \Psi}_{N_b}
    \\\nonumber
    &\quad \times {\rm diag}\big({\pmb \Psi}_{N_b}^H\bW_j\bW_j^H{\pmb \Psi}_{N_b}\big){\pmb \Psi}_{N_b}^H\bg_{j,i,u}\!(k) 
    \\\nonumber
    &=\alpha(1-\alpha)\sum\nolimits_{i',u',n,j}\xi_{i',u'}(n)\bg_{j,i',u'}^H\!(n){\pmb \Psi}_{N_b} 
    \\\nonumber
    &\quad\times {\rm diag}\big({\pmb \Psi}_{N_b}^H\bW_j\bW_j^H{\pmb \Psi}_{N_b}\big){\pmb \Psi}_{N_b}^H\bg_{j,i',u'}(n)
    \\ \nonumber
    &=\alpha(1-\alpha)\sum\nolimits_{j,v,\ell}\bw_{j,v}^H(\ell){\pmb \Psi}_{N_b}(\ell)  
    \\\nonumber
    &\quad \times {\rm diag}\left({\pmb \Psi}_{N_b}^H\bG_j\,{\Xi}\,\bG_j^H{\pmb \Psi}_{N_b}\right){\pmb \Psi}_{N_b}^H(\ell)\bw_{j,v}(\ell)
    \\\nonumber  
    &\stackrel{(b)}= \alpha(1-\alpha)\sum\nolimits_{j,v,\ell}\bw_{j,v}^H(\ell){\pmb \Psi}_{N_b}(\ell)
    \\\label{eq:QN_reform_ofdm} 
    &\quad\times {\rm diag}\big({\pmb \Psi}_{N_b}^H\bg_{j,i,u}\!(k)\bg_{j,i,u}^H\!(k)^H{\pmb \Psi}_{N_b}\big){\pmb \Psi}_{N_b}^H(\ell)\bw_{j,v}^H(\ell), 
\end{align}
where $(a)$ is obtained by substituting \eqref{eq:Cqq_dl_ofdm} into \eqref{eq:quantizationnoise} and
$(b)$ comes from the fact that $\bG_j\,{{\Xi}}\,\bG_j^H = \bg_{j,i,u}\!(k)\bg_{j,i,u}^H(k)$ since ${\Xi}$ can activate the $(kN_u+u)$th column of $\bG_{j}$ only.
	
Accordingly, the active DL SQINR constraint is rewritten as follows:
\begin{align}
    \nonumber
    \sigma^2=&\frac{\alpha^2}{\gamma_{i,u}(k)} |{\bf g}_{i,i,u}^H(k){\bf w}_{i,u}(k)|^2 
    \\\nonumber
    &-\alpha^2 \sum_{(j,v) \neq (i,u)}^{N_c, N_u}\!|\bg_{j,i,u}^H(k)\bw_{j,v}(k)|^2  -{\rm Q}_{i,u}(k) 
    \\\nonumber
    \stackrel{(a)}=&  \frac{\alpha^2}{\gamma_{i,u}} |{\bf g}_{i,i,u}^H(k)\bff_{i,u}^{\sf MMSE}(k)|^2 \tau_{i,u}(k) 
    \\\nonumber
    &-  \alpha^2 \sum_{(j,v) \neq (i,u)} {| {\bf g}_{j,i,u}(k) \bff_{j,v}^{\sf MMSE}(k)^H |^2} \tau_{j,v}(k) 
    \\ \nonumber
    &-\alpha(1-\alpha)\!\sum\nolimits_{j,v,\ell}\!\tau_{j,v}(\ell)\bff_{j,v}^{\sf MMSE}(\ell)^H{\pmb \Psi}_{N_b}(\ell)
    \\ \label{eq:dl_active}
    &\times {\rm diag}\big({\pmb \Psi}_{N_b}^H\bg_{j,i,u}\!(k)\bg_{j,i,u}^H\!(k){\pmb \Psi}_{N_b}\big){\pmb \Psi}_{N_b}^H(\ell)\bff_{j,v}^{\sf MMSE}(\ell)^H,
\end{align} 
for all $i$, $u$, $k$ where $(a)$ is from \eqref{eq:QN_reform_ofdm} and $\bw_{i,u}(k) = \sqrt{\tau_{i,u}(k)}\bff_{i,u}^{\sf MMSE}(k)$. 
Representing \eqref{eq:dl_active} for all $i,u,k$ gives  $\sigma^2{\bf 1} = {\bSigma}{\btau} $, thereby having $\tau_{i,u}(k)$ by solving ${\btau}=\sigma^2 {\bSigma}^{-1}{\bf 1}$.
\qed

\section{Proof of Corollary~\ref{cor:solution_ofdm}}
\label{appx:solution_ofdm}
\eqref{eq:solution_ofdm} is obtained by setting the Lagrangian in \eqref{eq:derivative} to zero, solving for $\lambda_{i,u}(k)$, and replacing $\frac{\alpha}{K}\tilde\bD_{i}$ with $\bD_i$. 
Thus, the solution of $\lambda_{i,u}(k)$ satisfies the stationary condition, and we observe that the UL SQINR constraint in \eqref{eq:per_antenna_const}
is active at the solution satisfying the complementary slackness condition. Therefore, \eqref{eq:solution_ofdm} is optimal solution of the UL OFDM problem.
\qed

\section{Proof of Corollary~\ref{cor:subgradient}}
\label{appx:subgradient}
Using the BF duality between UL and DL problems shown in \cite{dahrouj2010coordinated} and \cite{choi2020quantized},  $f\!\left(\bD_i\right)$ in \eqref{eq:subproblem} for fixed $\bD_i$ can be transformed to the following DL BF problem:
	\begin{align}
		\label{eq:subproblem1}
		f\!\left(\bD_i\right)=&\min_{\bw_{i,u}(k)} \;\;  \sum\nolimits_{u,k}\bw_{i,u}^H(k)\bD_i\bw_{i,u}(k)\\
		\nonumber
		&{\text{ \rm subject to}}  \;\; \Gamma_{i,u}(k) \geq \gamma_{i,u}(k) \;\;\forall i,u,k.
	\end{align}
	We introduce two arbitrary diagonal covariance matrices $\bD_i$ and $\bD'_i$ whose associated optimal solution of \eqref{eq:subproblem1} is denoted as $\bw_{i,u}(k)$ and $\bw'_{i,u}(k)$, respectively.
	We can then discover the following inequality on the objective function:
	\begin{align}
		&f\!\left(\bD'_i\right)-f\!\left(\bD_i\right) \nonumber\\
  &=\sum\nolimits_{u,k}\bw_{i,u}'^{H}(k)\bD_i'\bw_{i,u}'(k)-\sum\nolimits_{u,k}\bw_{i,u}^H(k)\bD_i\bw_{i,u}(k) \nonumber
		\\
		&\stackrel{(a)}{\leq}\sum\nolimits_{u,k}\bw_{i,u}^H(k)\bD_i'\bw_{i,u}(k)-\sum\nolimits_{u,k}\bw_{i,u}^H(k)\bD_i\bw_{i,u}(k) \nonumber
		\\
		&\stackrel{(b)}={\rm tr}\Big({\rm diag}\Big(\sum\nolimits_{u,k}\bw_{i,k}(k)\bw_{i,k}^H(k)\Big)\left(\bD_i'-\bD_i\right)\Big),
	\end{align}
	where $(a)$ holds because associating $\bD_i'$ with $\bw_{i,u}(k)$ cannot decrease the objective function from $\bw'_{i,u}(k)$, and $(b)$ follows because $(\bD_i'-\bD_i)$ is diagonal.
	From the definition of a subgradient: $\bA$ is a subgradient if $f(\bD_{i}') \leq f(\bD_i)  + {\rm tr}(\bA(\bD_i'-\bD_i))$, \eqref{eq:subgradient} is a subgradient of $f(\bD_i)$. 
 \qed

\section{Proof of Corollary~\ref{cor:convergence}}
\label{appx:convergence}

        The proof is based on the standard function approach presented in \cite{choi2020quantized,dahrouj2010coordinated,wiesel2005linear,yates1995framework}.
    Let us represent \eqref{eq:solution_ofdm} as $\lambda_{i,u}^{(n+1)}(k) = \mathcal{F}_{i,u,k}(\bLambda^{(n)})$. 
    We show that $\mathcal{F}_{i,u,k}(\bLambda)$ is a standard function which satisfies the followings:
    \begin{itemize}
        \item (Positivity) If $\lambda_{i,u}(k) \geq 0$ $\forall i,u,k$, then $\mathcal{F}_{i,u,k}({\bLambda}) > 0$.
        \begin{proof}
        By limiting $\lambda_{i,u}^{(n+1)}(k)$ to be non-negative, we have ${\bf K}_{i,k}({\pmb \Lambda}) \succ 0$; hence ${\bf K}_{i,k}^{-1}({\pmb \Lambda}) \succ 0$. This makes the denominator of \eqref{eq:solution_ofdm} strictly positive.
        \end{proof}

        \item (Monotonicity) If $\lambda_{i,u}(k) \geq \lambda_{i,u}'(k) \;\forall i,u,k$, then $\mathcal{F}_{i,u,k}({\bLambda}) \geq \mathcal{F}_{i,u,k}({\bLambda}')$.  
       \begin{proof}
            We define the signal and quantization noise part of the covariance matrix as
            \begin{align}
                &\bR({\pmb \Lambda})=\alpha \!\sum_{j,v} \!\lambda_{j,v}(k){\bf g}_{i,j,v}(k){\bf g}_{i,j,v}^H(k) \nonumber\\
                &\,\,\, +  (1\!-\!\alpha){\pmb \Psi}_{\!N_b}\!(k){\rm diag}\Big(\!{\pmb \Psi}_{N_b}^H\bG_i {\pmb \Lambda} \bG_i^H{\pmb \Psi}_{\!N_b}\!\Big){\pmb \Psi}_{\!N_b}^H\!(k)
            \end{align}
            Since ${\pmb \Lambda}$ is a diagonal matrix, the covariance matrix of the received signal is rewritten as
            \begin{align}
            {\bf K}_{i,k}({\pmb \Lambda}) = {\bf D}_{i}  + \bR({\pmb \Lambda}) = {\bf D}_{i} + \bR({\pmb \Lambda}')  + \bR({\pmb \Lambda}-{\pmb \Lambda}').
        \end{align}
            Due to the condition of $\lambda_{i,u}(k) \geq \lambda_{i,u}'(k)$, we can notice that $\bR(({\pmb \Lambda}-{\pmb \Lambda}')) \succeq 0$. 
            In addition, we know ${\bf D}_{i}  + \bR({\pmb \Lambda}) \succeq 0$ and ${\bf g}_{i,i,u}^H(k)$ is in the range of ${\bf D}_{i}  + \bR({\pmb \Lambda}')$ .
            We can recognize
            \begin{align}
                \nonumber
                &\mathcal{F}_{i,u,k}({\bLambda})= 
                \\\nonumber
                 &\frac{1}{\alpha \!\left(\!1\!+\!\frac{1}{\gamma_{i,u}\!(k)}\!\right)\!{\bf g}_{i,i,u}^H\!(k)\!\left({\bf D}_{i}\! +\! \bR({\pmb \Lambda}') \! +\! \bR({\pmb \Lambda}\!-\!{\pmb \Lambda}')\right)^{-1}\!\!\!\!{\bf g}_{i,i,u}\!(k)} 
                \\\nonumber
                &\geq \frac{1}{\alpha \left(1+\frac{1}{\gamma_{i,u}(k)}\right){\bf g}_{i,i,u}^H(k)\left({\bf D}_{i} + \bR({\pmb \Lambda}')\right)^{-1}{\bf g}_{i,i,u}(k)} 
                \\\nonumber
                &=\mathcal{F}_{i,u,k}({\bLambda'})
            \end{align} 
            by Proposition 4 of \cite{wiesel2005linear}.
        \end{proof}
        \item (Scalability) For $\rho > 1$, $\rho \mathcal{F}_{i,u,k}({\bLambda}) > \mathcal{F}_{i,u,k}(\rho{\bLambda})$.
       \begin{proof}
        We have $(\rho-1){\bf D}_{i}$ is still a semidefinite matrix. 
        By Proposition 4 of \cite{wiesel2005linear}, the following relationship holds:
        \begin{align}
            \nonumber
            &\rho\mathcal{F}_{i,u,k}({\bLambda}) 
            \\\nonumber
            &= \frac{1}{\alpha \left(1+\frac{1}{\gamma_{i,u}(k)}\right){\bf g}_{i,i,u}^H(k)\left(\rho{\bf D}_{i} + \rho\bR({\pmb \Lambda})\right)^{-1}{\bf g}_{i,i,u}(k)}
            \\\nonumber
            &\geq \frac{1}{\alpha \left(1+\frac{1}{\gamma_{i,u}(k)}\right){\bf g}_{i,i,u}^H(k)\left({\bf D}_{i} +\rho\bR({\pmb \Lambda}) \right)^{-1}{\bf g}_{i,i,u}(k)} 
            \\\nonumber
            &=\mathcal{F}_{i,u,k}({\rho\bLambda})
        \end{align} 
        We can realize that the equality condition holds when $\rho=1$, which is not feasible.
        \end{proof}
    \end{itemize}

\end{appendices}

 \bibliographystyle{IEEEtran}
 \bibliography{CoMP_ADCs.bib}

\end{document}